\documentclass[11pt]{article}

\usepackage{amssymb,amsmath,amsthm, amsfonts}
\usepackage{latexsym}
\usepackage{verbatim}
\usepackage{fullpage}
\usepackage{epsfig}

\usepackage{hyperref}
\usepackage{color}
\usepackage{braket}
\usepackage{stmaryrd}
\usepackage{graphicx}
\usepackage{hyperref}
\usepackage{float}
\usepackage{cleveref}
\usepackage{comment}
\usepackage[linesnumbered,boxed]{algorithm2e}

\Crefname{algocf}{Algorithm}{Algorithms}

\newtheorem{definition}{Definition}[section]
\newtheorem{theorem}{Theorem}[section]
\newtheorem{lemma}[theorem]{Lemma}
\newtheorem{proposition}[theorem]{Proposition}
\newtheorem{corollary}[theorem]{Corollary}

\newcommand{\triplep}[3]{\mbox{$\bra{#1}#2\ket{#3}$}}

\newcommand{\Zed}{\mathbb{Z}}
\newcommand{\Z}{\Zed}
\newcommand{\Eff}{\mathbb{F}}
\newcommand{\F}{\Eff}
\newcommand{\Cee}{\mathbb{C}}

\newcommand{\pluseq}{\;\text{+=}\;}

\newcommand{\Ptime}{\mathsf{P}}

\newcommand{\BQP}{\mathsf{BQP}}

\newcommand{\NumP}{\mathsf{\#P}}

\newcommand{\gatem}[1]{\mathsf{#1}}

\newcommand{\mat}[1]{\mathbf{#1}}
\renewcommand{\vec}[1]{\mathbf{#1}}
\newcommand{\diag}{\mathit{diag}}

\newsavebox{\fmbox}
\newenvironment{fmquote}[1]
{\vspace{1ex}\begin{center}\begin{lrbox}{\fmbox}\begin{minipage}{#1}\centering\vspace{0.5ex}}
{\vspace{0.2ex}\end{minipage}\end{lrbox}\fbox{\usebox{\fmbox}}\end{center}\vspace{1ex}}

\begin{document}
\sloppy

\title{\Large\bf Stabilizer Circuits, Quadratic Forms, and Computing Matrix Rank
}
\author{Chaowen Guan\\
University at Buffalo (SUNY)\\
{\tt chaoweng@buffalo.edu}
\and
Kenneth W. Regan\\
University at Buffalo (SUNY)\\
{\tt regan@buffalo.edu}
}
\date{\today}
\maketitle

\begin{abstract}
We show that a form of strong simulation for $n$-qubit quantum stabilizer circuits $C$ is computable in $O(s + n^\omega)$ time, where $\omega$ is the exponent of matrix multiplication.  Solution counting for quadratic forms over $\Eff_2$ is also placed into $O(n^\omega)$ time.  This improves previous $O(n^3)$ bounds.  Our methods in fact show an $O(n^2)$-time reduction from matrix rank over $\Eff_2$ to computing $p = |\triplep{0^n}{C}{0^n}|^2$ (hence also to solution counting) and a converse reduction that is $O(s + n^2)$ except for matrix multiplications used to decide whether $p > 0$.  The current best-known worst-case time for matrix rank is $O(n^{\omega})$ over $\Eff_2$, indeed over any field, while $\omega$ is currently upper-bounded by $2.3728\dots$  Our methods draw on properties of classical quadratic forms over $\Zed_4$.
We study possible distributions of Feynman paths in the circuits and prove that the differences in $+1$ vs.\ $-1$ counts and $+i$ vs.\ $-i$ counts are always $0$ or a power of $2$.  Further properties of quantum graph states and connections to graph and matroid theory are discussed.
\end{abstract}

\section{Introduction}\label{introduction}

Consider the following algorithm for computing the rank $r$ of an $n \times n$ matrix $\mat{A}_0$ over the field $\Eff_2$:
\begin{enumerate}
\item
Form the symmetric block matrix $\mat{A} = \begin{bmatrix} 0 & \mat{A}_0 \\ \mat{A}_0^\top & 0 \end{bmatrix}$.
\item
Form the quantum \emph{graph state} circuit $C_{\mat{A}}$ for the bipartite graph with adjacency matrix $\mat{A}$.
\item
Calculate $p =$ the quantum probability that $C_{\mat{A}}(0^{2n}) = 0^{2n}$.  The bipartite case assures $p > 0$.
\item
Output $r = \log_2(1/\sqrt{p})$.
\end{enumerate}

\noindent
All steps except 3 take $O(n^2)$ time.  Hence, for dense matrices, this is a linear-time reduction from $r$ to $p$.  In the converse direction, we will show the following algorithm for computing the amplitude $\triplep{0^n}{C}{0^n}$ for any $n$-qubit quantum \emph{stabilizer circuit} $C$:

\begin{enumerate}
\item
Convert $C$ to a classical quadratic form $q_C$ over $\Zed_4$ that retains all quantum properties of $C$.
\item
Take the matrix $\mat{A}$ of $q_C$ over $\Zed_4$ and associate a canonical $n \times n$ matrix $\mat{B}$ over $\Eff_2$ to it.  
\item
Compute the decomposition $\mat{B} = \mat{P}\mat{L}\mat{D}\mat{L}^\top \mat{P}^\top$ over $\Eff_2$ where $\mat{P}$ is a permutation matrix, $\mat{L}$ is lower-triangular, and $\mat{D}$ is block-diagonal with blocks that are either $1 \times 1$ or $2 \times 2$.
\item
Take $\mat{L}^{-1}$ over $\Eff_2$ but compute $\mat{D}' = \mat{L}^{-1}\mat{P}^\top \mat{A} \mat{P} (\mat{L}^{-1})^\top$ over $\Zed_4$.  (Note $\mat{P}^\top = \mat{P}^{-1}$.)  If any diagonal $1 \times 1$ block of $\mat{D}$ has become $2$ in $\mat{D}'$, output $\triplep{0^n}{C}{0^n} = 0$.  Else, $\triplep{0^n}{C}{0^n}$ is nonzero and is obtained by a simple 
$O(n)$-time recursion.
\end{enumerate}

\noindent
Here step 1 from \cite{RCG18} takes time linear in the number $s$ of quantum gates in $C$, which for standard-basis inputs can be bounded above by $O(n^2/\log n)$ with $O(n)$ quantum Hadamard gates \cite{AaGo04}.  Step~3 is computable in $O(n^\omega)$ time by \cite{DuPe18}, where $\omega$ is the exponent of matrix multiplication and is at most $n^{2.372865}$ \cite{Sto10,Wil12,LeGall14}.  This is also the best known time for computing $n \times n$ matrix rank over any field and for the particular inverses and products in step~4 as well (see \cite{CKL13}).  However, when $\triplep{0^n}{C}{0^n} \neq 0$ we show that its absolute value is computable quickly from $r$ alone after step 2.

Graph-state circuits and the larger but equivalent class of stabilizer (aka.\ Clifford) circuits are commonly quoted as simulatable in $O(n^2)$ time but this applies only with a bounded number of single-qubit measurements \cite{AaGo04,AnBr06} (see also \cite{GaMa13,GMC14,GaMa15}).
Computing the probability $p = |\triplep{0^n}{C}{0^n}|^2$ is classed as a form of \emph{strong simulation} by \cite{JovdN14} and is representative of the tasks designated $\mathsf{STR}(n)$ in \cite{Koh17} for standard-basis inputs.  The best times stated for computing $p$ in the above-cited papers are $O(n^3)$ for the general $n$-qubit case.  Thus our results improve the asymptotic running time as well as show a near-tight relationship to the task of computing matrix rank that seems not to be noticed in these papers.  They also improve the $O(n^3)$-time algorithm for solution counting of quadratic forms over $\Zed_2$, as given in \cite{EhKa90}, to  $O(n^\omega)$.  (As is common in the literature, we slur the distinction between $\omega$ as an infimum and as an upper bound---the latter usage should properly say ``in time $n^{\omega + o(1)}$'' or similar.  Our machine model is implicitly a RAM that can handle $\log n$-sized words in unit time; for other models we can say the time bounds ignore log factors.)

\bigskip
\begin{theorem}[main]\label{mainthm}
\begin{itemize}
\item[(a)]
Strong simulation of $n$-qubit stabilizer circuits of size $s$ with $h$ Hadamard gates (or other nondeterministic single-qubit gates) on standard-basis inputs is in time $O(s + n + h^{\omega})$ where $2 \leq \omega < 2.3729$.  This works for amplitude as well as probability.
\item[(b)]
Computing $n \times n$ matrix rank is linear-time equivalent to computing the probability $p$ (for circuits where $h = Theta(n)$ and $s = O(n^2))$ on the promise that $p$ is positive, and equivalent to computing $p$ on the narrower promise that the graphs underlying the circuits are bipartite.
\end{itemize}
\end{theorem}

\bigskip
\noindent
In view of the normal form of \cite{AaGo04} and in practice, the restriction on $h$ and $s$ in (b) is highly reasonable.  The ``promise'' formulation of (b) is ignorable in the direction from the rank $r$ to $p$, but not from $p$ to $r$.  The sense of the latter direction is that if rank for dense matrices comes to have a lesser time $t(n)$ with $n^2 \leq t(n) < n^{\omega}$ than matrix multiplication, then computing $p$ correctly in cases where $p > 0$ will have exactly the same time $t(n)$, whereas computing $p$ in all cases might remain in $n^\omega$ time.  We do not have a reduction from matrix multiplication itself (over $\Eff_2$) to strong simulation, hence our results do not imply an asymptotic equivalence between those.  To be sure, we note as a practical caveat that among the known sub-cubic algorithms for matrix multiplication, only Strassen's original one \cite{Str69}, which runs in time $O(n^{2.81})$, is considered competitive for problem sizes in the range of thousands of qubits that are addressed concretely in the above-cited papers.

The connections used in our proof run through the real-time conversion of quantum circuits $C$ to ``phase polynomials'' $q_C$ over $\Zed_K$ for $K = 2^k$, $k \geq 1$ in \cite{ReCh09,RCG18}, which extended results by \cite{DHHMNO04} for $k = 1$, and the analysis of quadratic forms over $\Zed_4$ by Schmidt \cite{Schm09} drawing on \cite{Alb38,Bro72}.  In the case of graph-state circuits and \emph{stabilizer circuits} more generally, $q_C$ becomes a classical quadratic form over $\Zed_4$, as treated also in \cite{CGW18}.  Our approach is related to ones involving Gauss sums \cite{BvDR08,CCLL10,CGW18,BuKo18} but exploits the availability of normal forms.  For bipartite $\mat{A}$ as above, it further devolves into a quadratic form $q'_C$ over $\Eff_2$ that is \emph{alternating} (as defined below) plus an ancillary vector $v$.  A linear change in basis---which also sends $v$ to a vector $w$ but leaves the probability computation unaffected---gives over $\Zed_4$ the normal form
\begin{equation}\label{NF}
q'_C = y_1 y_2 + y_3 y_4 + \cdots + y_{2g-1} y_{2g} + \sum_{j=1}^{n} 2y_j w_k.
\end{equation}
Here the rank $r$ must be even and $g = r/2$.  This corresponds to block-diagonal matrices $\mat{D}$ with $g$-many $2 \times 2$ blocks as produced by \cite{DuPe18}, together with $1 \times 1$ blocks coming from $w$.  The $1 \times 1$ blocks matter most for $j > r$.  The matrix $\mat{D}'$ over $\Zed_4$ may no longer be block-diagonal but its diagonal reveals $w$.

Provided the terms in $w$ do not cause global cancellation, equation (\ref{NF}) will yield $p$ from $r$ in an invertible manner, without needing to compute the change in basis.
Let $N_c(q)$ stand for the number of arguments $x \in \{0,1\}^n$ giving $q(x) = c \pmod{4}$ for $c = 0,1,2,3$.  Along the way to our main theorem, we prove that for any classical quadratic form $q$ over $\Zed_4$, the differences $|N_0(q) - N_2(q)|$ and $|N_1(q) - N_3(q)|$ are either zero or a power of $2$.  This resolves the effects of the ``$w$'' part of the normal form (\ref{NF}) for the alternating case in particular.  Sections~\ref{stabilizer} and~\ref{sec:properties} cover stabilizer circuits and quadratic forms before sections~\ref{sec:rank} proves part (a) of Theorem~\ref{mainthm} and the rank-to-strong-simulation direction of part (b).


The other direction---which was the original goal---requires computing $r$ plus information about $w$.  The datum needed is whether a $\Zed_4$ vector corresponding to $w$ has an entry of value $2$, or equivalently, whether the graph underlying $C$ belongs to a family we call ``net-zero'' graphs.  Section~\ref{sec:self-dual} analyzes a concept of ``self-dual'' quadratic forms yielded by the probability computation and reduces from the general to the alternating case---which effectively strips phase gates from the circuit and self-loops from the graph---and finishes the proof of part (b) of Theorem~\ref{mainthm}.  

A concluding section Section~\ref{conclusions} discusses the ``net-zero'' graphs and observes that the amplitude function $a(G) = \triplep{0^n}{C_G}{0^n}$ is a generalized Tutte invariant per \cite{OxWh93,Nob06}.  It then
contrasts the integral but non-classical quadratic forms which arise from adding the controlled-phase gate to the stabilizer gates to form a universal gate set.  Finally we raise possible implications of this work for solution counting and for graph theory.

\section{Quantum Stabilizer Circuits}\label{stabilizer}

A fundamental problem in quantum computing is whether all quantum circuits of $s$ gates acting on $n$ qubits can be simulated in time polynomial in $s$ and $n$.  A quantum circuit $C$ effects a unitary linear transformation on $\Cee^N$ where $N = 2^n$.  A fixed basis of $\Cee^N$ is identified with $\{0,1\}^n$.  The circuit is a sequence of \emph{gates} $g$, each of which effects a transformation $U_g$ that acts on some $k$ of the qubits.
The gate $g$ can be represented by a $2^k \times 2^k$ unitary \emph{gate matrix} $\mat{M}_g$ and the subset $S_g$ of qubits acted on.

A salient subclass of quantum circuits that have a deterministic polynomial-time simulation are \emph{stabilizer circuits}.  They can be generated by the following three gate matrices $\mat{M}_g$:
\[
\gatem{H} = \frac{1}{\sqrt{2}}\begin{bmatrix} 1 & 1 \\ 1 & -1 \end{bmatrix},\quad
\gatem{S} = \begin{bmatrix} 1 & 0 \\ 0 & i \end{bmatrix},\quad
\gatem{CZ} = \begin{bmatrix} 1 & 0 & 0 & 0 \\ 0 & 1 & 0 & 0 \\ 0 & 0 & 1 & 0 \\ 0 & 0 & 0 & -1 \end{bmatrix}.
\]
Their extensions to act on $\Cee^N$ by tensor product with the identity on the qubits outside $S_g$ generate the $n$-qubit \emph{Clifford group}, so these are also called \emph{Clifford circuits}.  The original polynomial-time algorithm by Gottesman and Knill \cite{Got98} involved Gaussian elimination and so ran for all intents and purposes in order-of $n^3$ time.  Aaronson and Gottesman \cite{AaGo04} improved this to $O(n^2)$ time with a tableau method and also showed that every stabilizer circuit has an equivalent one with $O(n^2/\log n)$ gates.  Anders and Briegel \cite{AnBr06} improved the running time concretely and for circuits of size $s = o(n^2)$ using a graph-state representation, as we will also do.  Dehaene and De Moor \cite{DeDM03} described quantum states produced by stabilizer circuits via linear and quadratic forms over $\Eff_2$ in ways simplified and extended by van den Nest \cite{vdNest09}.  

We seek even simpler and faster methods that lend themselves to further algorithmic properties, such as quick update when changes are made to $C$ in the sense of ``dynamic algorithms.''  We employ the theory of classical quadratic forms over $\Zed_4$ as developed by Schmidt \cite{Schm09} and more recently by Cai, Guo, and Williams \cite{CGW18}.  The quadratic forms are built using the real-time algorithm of \cite{ReCh09,RCG18} for computing what we call the \emph{additive partition polynomials} $q_C$ for quantum circuits $C$ that meet a mild ``balance'' condition.  Related works involving low-degree polynomials and counting complexity include \cite{BvDR08,BJS10,Mon17,KPS17}.

The polynomial $q_C$ has variables $x_1,\dots,x_n$ corresponding to binary input values, $z_1,\dots,z_n$ for the binary output values, and $y_1,\dots,y_h$ representing nondeterminism from Hadamard (and possibly other) gates.  For any $a,b \in \{0,1\}^n$, letting $q_{ab}$ denote $q$ with those values substituted for the $x_i$ and $z_j$ variables, we have for some $R > 0$:
\begin{equation}\label{amplitude}
\triplep{b}{C}{a} = \frac{1}{R} \sum_{y \in \{0,1\}^h} \omega^{q_{ab}(y)},
\end{equation}
where $\omega$ is a $K$-th root of unity such that all phases produced by the circuit are powers of $\omega$.  Stabilizer circuits give $K = 4$ so that the powers in this exponential sum belong to $\Zed_4$.  Generally $R = 2^{h/2}$ but its value is reduced if some nondeterministic $y_j$ variables are forced to equal outputs.

The rules for calculating $q$ are straightforward.  Initially $q = 0$ and each qubit line $i$ has its current \emph{annotation} $u_i$ defined by $u_i = x_i$.  In general, let $u_i$ stand for the current annotation of line $i$, and let $y_1,\dots,y_{\ell-1}$ be the nondeterministic variables allocated thus far.

\begin{itemize}
\item
Hadamard gate on line $i$: Allocate a new variable $y_\ell$, do $q \pluseq 2u_i y_\ell$, and reassign $u_i$ to be $y_\ell$.
\item
Phase gate $\gatem{S}$ on line $i$: $q \pluseq u_i$, $u_i$ unchanged.
\item
$\gatem{CZ}$ gate on lines $i$ and $j$: $q \pluseq 2u_i u_j$, no other change.
\item
At the end of each qubit line $i$, we can identify $z_i$ with the variable last denoted by $u_i$.
\end{itemize}

\noindent
Since we are concerned only with $0,1$ as arguments, we can also do $q \pluseq u_i^2$ in the case of $\gatem{S}$, thus making all terms homogeneously quadratic.  The \emph{conjugate polynomial} $q^*$ does $q^* \pluseq 3u_i^2$ instead, but does the same as $q$ for $\gatem{H}$ and $\gatem{CZ}$.

We mention some other Clifford gates and their rules for completeness.  The first three (plus the identity $\gatem{I}$) are the \emph{Pauli gates}:
\[
\gatem{X} = \begin{bmatrix} 0 & 1 \\ 1 & 0 \end{bmatrix},\quad
\gatem{Y} = \begin{bmatrix} 0 & -i \\ i & 0 \end{bmatrix},\quad
\gatem{Z} = \begin{bmatrix} 1 & 0 \\ 0 & -1 \end{bmatrix},\quad
\gatem{CNOT} = \begin{bmatrix} 1 & 0 & 0 & 0 \\ 0 & 1 & 0 & 0 \\ 0 & 0 & 0 & 1 \\ 0 & 0 & 1 & 0 \end{bmatrix}.\quad
\]
\begin{itemize}
\item
$\gatem{X}$: Change $u_i$ to $1 - u_i$; no other change.
\item
$\gatem{Z}$: $q \pluseq 2u_i$; same for $q^*$; no other change.
\item
$\gatem{Y}$: Treat via $Y = i\gatem{X}\gatem{Z}$ and ignore the global scalar $i$.
\item
$\gatem{CNOT}$ gate on lines $i$ and $j$: Change the target $u_j$ to $u_i + u_j - 2 u_i u_j$; no other action.
\end{itemize}

\noindent
Note that $\gatem{X} = \gatem{H}\gatem{Z}\gatem{H}$, and similarly, $\gatem{CNOT} = (\gatem{I} \otimes \gatem{H})\gatem{CZ}(\gatem{I} \otimes \gatem{H})$, so we could have composed the previous rules for those gates.  Changing the annotations $u_i$ and $u_j$ (in the respective cases), however, avoids introducing new nondeterministic variables.

The annotation $u_j$ becomes quadratic in the case of $\gatem{CNOT}$, but the degree does not rise any higher:  In rules where $u_i$ is multiplied, the multiplier contains a factor $2$ which cancels the $2u_i u_j$ modulo $4$.  The last subtlety is what happens when an annotation that is not a single variable is to be equated with a variable $z_j$.  If it has the form $u_i + u_j - 2 u_i u_j$ then we add to $q$ the term $2w(u_i + u_j - 2u_i u_j - z_j) = 2wu_i + 2wu_j + 2wz_j \pmod{4}$ where $w$ is a fresh variable.  For binary values from the standard basis, if $z_j$ does not equal the XOR of $u_i$ and $u_j$ then the added term reduces to $2w$.  Because $w$ appears nowhere else, assignments with $w = 0$ and those with $w = 1$ will globally cancel in (\ref{amplitude}).  Thus only cases with $z_j = u_i \oplus u_j$ contribute.  This proves

\bigskip
\begin{theorem}[\cite{RCG18}]
When $C$ is a stabilizer circuit, the polynomial $q$ in (\ref{amplitude}) becomes a quadratic form over $\Zed_4$ in which all terms involving two variables have coefficient $2$.
\end{theorem}

\bigskip
\noindent
Such forms are called \emph{classical}, reflecting the historical definition of a quadratic form as given by $x^\top \mat{A} x$ for some integer $n \times n$ matrix $A$ that is symmetric---so that all cross terms have even coefficients.  For $\Zed_4$ they coincide with those called \emph{affine} in \cite{CLX14,CGW18}.
For contrast, we note the effect of using a controlled-phase gate:
\[
\gatem{CS} = \begin{bmatrix} 1 & 0 & 0 & 0 \\ 0 & 1 & 0 & 0 \\ 0 & 0 & 1 & 0 \\ 0 & 0 & 0 & i \end{bmatrix}.
\]
If $i,j$ are the qubit lines involved, the rule is to do $q \pluseq u_i u_j$ and $q^* \pluseq 3u_i u_j$.  The $\gatem{CS}$ gate is not a Clifford gate and its inclusion creates a universal gate set.  Nevertheless, (\ref{amplitude}) still holds, as does the following:

\bigskip
\begin{theorem}[\cite{RCG18}]\label{probability}
For any Clifford+$\gatem{CS}$ circuit $C$, input $x$, outcome $z$, and letting $V$ stand for the set of unfixed variables overall,

\begin{eqnarray}\label{probsim}
\Pr[C(x) = z] &= \frac{1}{R^2}[&|\{v,v'\in\{0,1\}^{|V|} : q_{x,z}(v) + q^*_{x,z}(v') = 0\}| \\
&& - \;\;\;|\{v,v'\in\{0,1\}^{|V|} : q_{x,z}(v) + q^*_{x,z}(v') = 2\}|\quad],
\end{eqnarray}
with $R$ as in (\ref{amplitude}) and values modulo $4$.  Moreover, for any set $Z$ of binary outcomes defined by fixing a subset of the variables to a value $z'$, $\Pr[C(x) \in Z]$ is given by an analogous formula over pairs $v,v'$ of assignments that agree on the remaining $z_j$ variables involving $q_{x,z'}$ and $q^*_{x,z'}$.
\end{theorem}

\bigskip
\noindent
Using the ``$w_j$'' variables as above to equate outputs makes $R = 2^h$.
As we will note in Lemma~\ref{self-dual-basics}(d) in section~\ref{sec:self-dual}, it is unnecessary to put absolute-value bars on the difference in (\ref{probsim}).  There we will symmetrize the roles of $x$ and $z$ and address the general topic of a ``self-dual'' (classical) quadratic form.  We further note the remarkably fine-cut ``dichotomy'' that although counting solutions to $q(v) = 0$ over $v \in \Zed_4^r$ is in polynomial time for any quadratic form, counting them over $v \in \{0,1\}^r$ is $\NumP$-complete for a non-classical quadratic form \cite{CLX14}.  Classical quadratic forms are indifferent between $0$ and $2$ as arguments, likewise $1$ versus $3$, because $2^2 = 0$, $3^2 = 1$, and $2\cdot 1 = 2\cdot 3 = 2$ modulo $4$, so counting solutions over $\Zed_4^r$ and over $\{0,1\}^r$ is equivalent for them.

Let us bear in mind that since (\ref{amplitude}) computes all amplitudes, the polynomial $q = q_C$ includes all information about the quantum behavior of the circuit $C$.  Thus nothing is lost by manipulating (only) $q_C$.  As an application, we deduce the known fact that graph-state circuits are entirely representative of stabilizer circuits with $O(s+n)$ overhead.  Such circuits consist of:

\begin{itemize}
\item
An initial $n$-ary Walsh-Hadamard transform $\gatem{H}^{\otimes n}$, effected by placing one Hadamard gate at the start of each qubit line.
\item
For every edge $(i,j)$ in the given graph $G$, place a $\gatem{CZ}$ gate between lines $i$ and $j$.  Order does not matter because these operations commute.
\item
If $G$ has a self-loop at node $i$, place an $\gatem{S}$ gate there.
\item
A final $\gatem{H}^{\otimes n}$.
\end{itemize}

\begin{proposition}
There is a real-time procedure that given any $n$-qubit stabilizer circuit $C$ with $h$ Hadamard gates and $x,z \in \{0,1\}^n$ constructs a graph state circuit $C_G$ on $h$ qubits such that $\triplep{z}{C}{x} = \triplep{0^{h}}{C_G}{0^{h}}$.
\end{proposition}

\begin{proof}
Build $q_C$ in real time as above and substitute for $x$ and $z$.  This leaves $h$ variables $y_k$ from the Hadamard gates plus any $w_j$ variables that were employed.  Now define the graph $G$ to have an edge $(i,j)$ for every term $2 y_i y_j$ (or $2 y_i w_j$) in $q_C$, and $a$ self-loops at $i$ for every term $a y_i^2$, $a = 1,2,3$.  Note that the coefficients $a$ of the self-loop terms may arise from the substitutions for particular binary values of $x$ and $z$.  The corresponding graph-state circuit has inputs $x',z'$ of its own, but those are zeroed in forming $\triplep{0^{h}}{C_G}{0^{h}}$.  The leftover terms in $q_{C_G}$ are identical to those of $q_C$ after the substitution.
\end{proof}

If $h = \Theta(n)$ then the number of variables is linear in $n$.
Our original aim was to use this correspondence to be competitive with the above-cited $O(n^2)$ algorithms---and ones that improve then the graph is sparse---in the concrete sense of better leading constants and simplified cases.  For those algorithms previously not known to have time better than $O(n^3)$ or similar, our practical objective in what follows is not so much reducing the exponent to $\omega$ but rather to $O(n^2)$ time given knowledge of the rank $r$, for contexts where $r$ might be foreknown or well approximated.


\section{Properties of Classical Quadratic Forms Over $\Zed_4$}\label{sec:properties}

A \emph{classical quadratic form} $f$ in variables $\vec{x} = (x_1,\dots,x_n)$ is one induced by a symmetric $n \times n$ integer matrix $\mat{A}$ as
\begin{equation}\label{qform}
f(\vec{x}) = \vec{x}^\top \mat{A} \vec{x}.  
\end{equation}
This makes every coefficient of a cross term $x_i x_j$ even, and over $\Zed_4$ all nonzero cross terms have coefficient $2$.  Such a form over $\Zed_4$ treats arguments $0$ and $2$ the same, likewise $1$ and $3$, so we may regard it as a function of $\{0,1\}^n$ into $\Zed_4$.  Then we want to regard (\ref{qform}) as composed of matrix-vector operations over $\Eff_2$ plus some extra calculation to get the answer in $\Zed_4$ where $2,3$ as well as $0,1$ may be values.  

First note that by the symmetry, every off-diagonal entry of $\mat{A}$ may without loss of generality be $0$ or $1$.  Next, define a binary vector $\vec{v}$ by $v_j = 1$ if the $j$-th main-diagonal entry of $\mat{A}$ is $2$ or $3$, else $v_j = 0$.  Finally define a binary matrix $\mat{B}$ from $\mat{A}$ by
\begin{equation}\label{BfromA}
\mat{B} = \mat{A} - 2\diag(\vec{v}).
\end{equation}
Then we have
\begin{equation}\label{qformB}
f(\vec{x}) = \vec{x}^\top \mat{B} \vec{x} + 2\vec{x}^\top\cdot\vec{v}
\end{equation}
with calculation in $\Zed_4$.  The $\vec{x}^\top \mat{B} \vec{x}$ calculation is now valid in $\Eff_2$, however.  The quadratic form is \emph{alternating} if the main diagonal of $\mat{B}$ is all zero, else it is \emph{non-alternating}.  When $\mat{B}$ comes from or is regarded as the adjacency matrix of a graph, alternating means the graph is simple and undirected (as will hold in our reductions from rank using a simple bipartite graph) and non-alternating means the graph is undirected but with one or more self-loops.

We note the general development of this decomposition and associated concepts by Schmidt \cite{Schm09} in a way not wedded to the standard basis.  Since we fix 4 as the modulus throughout this section, we follow \cite{Schm09} in now using $K$ to denote $\{0, 1\}$ as a subset of $\Zed_4$, defining an operation $\oplus$ on $K$ by $a \oplus b := (a + b)^2$, and defining $V$ as an $n$-dimensional vector space ``over $K$'' noting that $(K,\oplus,\cdot)$ is the same as the field $\Eff_2$.  Then classical quadratic forms are equivalently defined as follows:

\bigskip
\begin{definition}[see \cite{Alb38,Schm09}]
A symmetric bilinear form on $V$ is a mapping $B : V \times V \rightarrow K$ that satisfies
\begin{enumerate}
\item
symmetry: $B(\vec{x}, \vec{y}) = B(\vec{y}, \vec{x})$;
\item bilinearity:
$B(\alpha \vec{x} \oplus \beta \vec{y}, \vec{z}) = \alpha B(\vec{x},\vec{z}) \oplus \beta B(\vec{y},\vec{z})$ for $\alpha, \beta \in K$.
\end{enumerate}
\end{definition}

\bigskip
\noindent
$B$ is \emph{alternating} if $B(\vec{x},\vec{x}) = 0$ for all $\vec{x} \in V$, else it is $\emph{non-alternating}$. Let $\Lambda = \{ \lambda_1, \cdots, \lambda_n \}$ be any basis for $V$ over $K$. Then $B$ is uniquely determined (relative to this basis) by the $n \times n$ matrix $\mat{B}$ with entries $b_{ij} = B(\lambda_i, \lambda_j)$. The \emph{rank} of $B$ is the rank of its matrix $\mat{B}$.

\bigskip
\begin{definition}[see \cite{Bro72,Schm09}]\label{dfn:quadratic}
A $\Z_4$-valued classical quadratic form is a mapping $f : V \rightarrow \Z_4$ that satisfies:
\begin{enumerate}
\item
$f(\alpha \vec{x}) = \alpha^2 f(\vec{x})$ for $\alpha \in K$;
\item
$f(\vec{x} \oplus \vec{y}) = f(\vec{x}) + f(\vec{y}) + 2 B(\vec{x}, \vec{y})$, where $B : V \times V \rightarrow K$ is a symmetric bilinear form.
\end{enumerate}
\end{definition}

\bigskip
\noindent
Then $f$ is \emph{alternating} if the associated bilinear form $B$ is alternating, \emph{non-alternating} otherwise, and its \emph{rank} $r$ is the rank of $B$.

\bigskip
\begin{proposition}[\cite{Schm09}]
There is a vector $\vec{v} \in K^n$ such that for all $\vec{x} \in K^n$ over the basis $\Lambda$,
\[
f(\vec{x}) = \vec{x}^\top \mat{B} \vec{x} + 2 \vec{x}^\top \cdot \vec{v} .
\]
\end{proposition}

\bigskip
\noindent
The point of dropping down to $\Eff_2$ is to leverage the notions of matrix similarity over $\Eff_2$ and the following theorem about changes of basis in $V$.
%
Over $\Eff_2$ the appropriate definition of $\mat{B}$ and $\mat{B}'$ being \emph{similar} (from \cite{Alb38}) is that there exists an invertible matrix $\mat{Q}$ such that $\mat{B}' = \mat{Q}^\top \mat{B} \mat{Q}$.  This preserves the property that similar matrices have the same rank. 
The notions of \emph{alternating} and \emph{non-alternating} are the same as given for the binary matrix $\mat{B}$ above, depending on whether the main diagonal of $\mat{B}$ is all zero or not.

\bigskip
\begin{theorem}[\cite{Alb38}]\label{thm:similar}
Let $\mat{A}$ be a $K$-valued $n \times n$ symmetric matrix of rank $r$.
\begin{itemize}
\item[(a)]
If $\mat{A}$ is alternating, then $\mat{A}$ has even rank and is similar to a matrix that has zeros everywhere except on the subdiagonal and the superdiagonal, which are $1010 \cdots 10100 \cdots 0$ with $r / 2$ ones.
\item[(b)]
If $\mat{A}$ is non-alternating, then $\mat{A}$ is similar to a diagonal matrix, whose main diagonal is of $r$-many ones.
\end{itemize}
\end{theorem}

\bigskip
With the representation of $f(\vec{x}) = \vec{x}^\top \mat{B} + 2 \vec{x}^\top \cdot \vec{v}$, the paper \cite{Schm09} uses this to define normal forms with regard to $\Zed_4$:

\bigskip
\begin{corollary}[\cite{Schm09}]\label{cor:similar}
%
%
%
Given a quadratic form $f$ of rank $r$ as above over the basis $\Lambda$, we can find
a basis $M = (\mu_1,\dots,\mu_n)$ for $V$ over $K$, mapping $\vec{x} = (x_1, \dots, x_n)$ over $\Lambda$ in $V$ to $\vec{y} = (y_1,\dots,y_n)$ such that:
\begin{itemize}
\item[(a)]
If $f$ is alternating, then
\[
f(\vec{y}) = 2 \sum^{r/2}_{j = 1} y_{2j - 1} y_{2j} + 2 \sum^{n}_{i = 1} w_i y_i,
 \]
for some $\vec{w} = (w_1, \cdots, w_n) \in K^n$.

\item[(b)]
If $f$ is non-alternating, then there is the equivalent linear form
\[
f(\vec{y}) = \sum^{r}_{j = 1} y_j + 2 \sum^{n}_{i = 1} w_i y_i,
\]
for some $\vec{w} = (w_1, \cdots, w_n) \in K^n$.
\end{itemize}
\end{corollary}

\bigskip
\noindent
Schmidt actually retains the symbols $\vec{x}$ and $\vec{v}$ in his statement but we have used $\vec{y}$ and $\vec{w}$ to indicate the change of basis.  Our analysis in the next section will, however, treat $\vec{y}$ as the standard basis, so the generic symbols $x_1,\dots,x_n$ will re-appear, and $w_1,\dots,w_n$ will just be ordinary 0-1 values.  This switch will be echoed in the next section in that once we substitute for the input qubit values $x_i$ and output values $z_j$ in the quadratic form $q_C$ from Section~\ref{stabilizer}, the actual variables of $q_C$ left over will be named $y_1,\dots,y_h$ where $h = O(n)$.  But to emphasize that the counting lemmas preceding the main results hold apart from the quantum context, we will revert to the standard symbols $x_1,\dots,x_n$ in their statements and proofs.

Now we reference \cite{DuPe18} to note some facts about matrix decompositions related to the above normal forms.  Note that the inverse of a non-singular lower triangular matrix is lower triangular.

\bigskip
\begin{lemma}\label{PLDLTPT}
\begin{itemize}
\item[(a)]
For every symmetric $n \times n$ matrix $\mat{B}$ over $\Eff_2$ there is a permutation matrix $\mat{P}$ such that the symmetric matrix $\mat{B}' = \mat{P}^\top \mat{B}\mat{P}$ has the decomposition $\mat{B}' = \mat{L}\mat{D}\mat{L}^\top$.  Here $\mat{L}$ is an $n \times n$ lower triangular matrix with unit diagonal and $\mat{D}$ is diagonal if $\mat{B}$ is non-alternating, else $\mat{D}$ is block-diagonal as described in Theorem~\ref{thm:similar}(a).
\item[(b)]
The matrix $\mat{D}$ in (a) is permutation-equivalent to any matrix $\mat{D}'$ fulfilling the corresponding case of Theorem~\ref{thm:similar} when applied to $\mat{B}'$ or to $\mat{B}$.
\item[(c)]
The matrix $\mat{D}$ in (a) is unique among LDU decompositions applied to $\mat{B}'$.
\item[(d)]
When $\mat{D}' = \mat{L}^{-1}\mat{P}^\top \mat{A}\mat{P} (\mat{L}^{-1})^\top$ is computed over $\Zed_4$ rather than $\Eff_2$, it may no longer be diagonal or block-diagonal, but it represents the same quadratic form with arguments in $V$ and values $\Zed_4$ in as in Corollary~\ref{cor:similar} over the new basis.  In both the alternating and non-alternating cases, the main diagonal of $\mat{D}'$ equals the main diagonal of $\mat{D}$ plus $2w$ where $w$ is the vector in Corollary~\ref{cor:similar}.
\end{itemize}
\end{lemma}

\begin{proof}
(a)~This is known and noted in \cite{DuPe18}.  A key point from Gaussian elimination is that if we alternate elementary matrices $\mat{L}_i$ that do elimination in the $i$th column of the lower triangle and swaps $\mat{P}_{j,k}$ of rows $j$ and $k$, then we can rewrite $\mat{P}_{j,k} \mat{L}_i$ where $j,k > i$ as $\mat{L}'_i \mat{P}_{j,k}$.  The matrix $\mat{L}'_i$ is obtained by interchanging the entries in rows $j$ and $k$ of column $i$ and those in positions $j$ and $k$ on the main diagonal.  (The latter is unnecessary when all diagonal entries are $1$) and is still lower-triangular.  Since each $\mat{L}'_i$ is still lower triangular and we can repeat the switch for further row swaps, we obtain the lower-triangular matrix formally designated as $\mat{L}^{-1}$ as the product of the $\mat{L}'_i$ and the matrix designated as $\mat{P}^\top$ as the product of all swaps.  Since $\mat{B}$ is symmetric, corresponding events on the right give $\mat{D} = \mat{L}^{-1}\mat{P}^\top\mat{B}\mat{P} (\mat{L}^{-1})^\top$ of the diagonal or block-diagonal forms stated in all of \cite{Alb38,Schm09,DuPe18}.

Part (b) follows simply because $\mat{D}$ and $\mat{D}'$ have the same rank and the same block-diagonal structure in the alternating case or diagonal structure in the non-alternating case).  The proof of (c), which is not strictly needed for our key point (d), is in the Appendix.

The point in (d) is that when computed over $\Zed_4$, $\mat{D}' = \mat{L}^{-1}\mat{P}^\top \mat{A} \mat{P} (\mat{L}^{-1})^\top$ represents the same quadratic form $f$ originally given by $\mat{A}$ in (\ref{qform}) but over the transformed basis that maps $\vec{x}$ to $\vec{y}$.  Thus
\begin{equation}
f(\vec{y}) = \vec{y}^\top \mat{D}' \vec{y} = \vec{y}^\top \mat{D} \vec{y} + 2 \sum_{i=1}^n y_i w_i.
\end{equation}
In the non-alternating case, this means any symmetric pairs $d'_{j,k},d'_{k,j}$ of off-diagonal elements of $\mat{D}'$ must sum to $0$ modulo $4$, and likewise off-diagonal elements in the alternating case apart from the block elements on the super-diagonal and sub-diagonal.  The diagonal must satisfy $d'_{j,j} = d_{j,j} + 2w_j \pmod{4}$ in either case.
\end{proof}


\bigskip
\noindent
Put more simply, the decomposition in \cite{DuPe18} is the same as that obtained in \cite{Schm09} following \cite{Alb38,Bro72}, so the normal forms for classical quadratic forms over $\Zed_4$ in the latter papers inherit the $O(n^\omega)$ time computability from \cite{DuPe18} working over $\Eff_2$.  For some examples, consider the alternating form $q(x_1,x_2,x_3) = 2x_1 x_2 + 2 x_1 x_3 + 2 x_2 x_3$.  It gives
\[
\mat{A} = \mat{B} = \begin{bmatrix} 0 & 1 & 1 \\ 1 & 0 & 1 \\ 1 & 1 & 0 \end{bmatrix},
\]
which is the adjacency matrix of the triangle graph.  Gaussian elimination begins by swapping row 1 and row 2, then no more swaps are needed.  So we have:
\[
\mat{P} = \mat{P}_{1,2} = \begin{bmatrix} 0 & 1 & 0 \\ 1 & 0 & 0 \\ 0 & 0 & 1 \end{bmatrix} = \mat{P}^\top, \quad
\mat{B}' = \mat{P}^\top \mat{B}\mat{P} = \mat{B},\quad\text{and}\quad \mat{L}^{-1} = \mat{L} = \begin{bmatrix} 1 & 0 & 0 \\ 0 & 1 & 0 \\ 1 & 1 & 1 \end{bmatrix}.
\]
This gives over $\Eff_2$,
\[
\mat{D} = \mat{L}\mat{B}\mat{L}^\top = \begin{bmatrix} 1 & 0 & 0 \\ 0 & 1 & 0 \\ 1 & 1 & 1 \end{bmatrix} \cdot \begin{bmatrix} 0 & 1 & 1 \\ 1 & 0 & 1 \\ 1 & 1 & 0 \end{bmatrix} \cdot \mat{L}^\top = \begin{bmatrix} 0 & 1 & 1 \\ 1 & 0 & 1 \\ 0 & 0 & 0 \end{bmatrix} \cdot \begin{bmatrix} 1 & 0 & 1 \\ 0 & 1 & 1 \\ 0 & 0 & 1 \end{bmatrix} = \begin{bmatrix} 0 & 1 & 0 \\ 1 & 0 & 0 \\ 0 & 0 & 0 \end{bmatrix}.
\]
But over $\Zed_4$, we get
\[
\mat{L}\mat{A} = \begin{bmatrix} 0 & 1 & 1 \\ 1 & 0 & 1 \\ 2 & 2 & 2 \end{bmatrix}, \text{ which times } \begin{bmatrix} 1 & 0 & 1 \\ 0 & 1 & 1 \\ 0 & 0 & 1 \end{bmatrix} = \begin{bmatrix} 0 & 1 & 2 \\ 1 & 0 & 2 \\ 2 & 2 & 2 \end{bmatrix} = \mat{D}' \equiv \begin{bmatrix} 0 & 1 & 0 \\ 1 & 0 & 0 \\ 0 & 0 & 2 \end{bmatrix}.
\]
The presence of a 2 in the lower-right corner of $\mat{D}'$, corresponding to a $1 \times 1$ block in the diagonal matrix $\mat{D}$, signals a cancellation in the 0-1 assignments $a \in K^n$ giving $q(a) = 0$ versus those giving $q(a) = 2$.  That is, $N_0(q) - N_2(q) = 0$.  In section~\ref{conclusions} we will call the simple triangle graph a ``net-zero'' graph.  

Now, however, let us define $q' = q + 2x_1^2$.  This corresponds to adding a self-loop at node 1 to the triangle graph.  This goes into the vector $\vec{v}$ and does not change $\mat{B}$ or the decomposition.  At the end, however, we first get that over $\Zed_4$, $\mat{A}' = \mat{P}^\top \mat{A} \mat{P}$ is no longer the same as $\mat{A}$: it moves the $2$ from the upper left corner to the center.  Then we get
\[
\mat{L}\mat{A}' = \begin{bmatrix} 0 & 1 & 1 \\ 1 & 2 & 1 \\ 2 & 0 & 2 \end{bmatrix}, \text{ which times } \begin{bmatrix} 1 & 0 & 1 \\ 0 & 1 & 1 \\ 0 & 0 & 1 \end{bmatrix} = \begin{bmatrix} 0 & 1 & 2 \\ 1 & 2 & 0 \\ 2 & 0 & 0 \end{bmatrix} = \mat{D}' \equiv \begin{bmatrix} 0 & 1 & 0 \\ 1 & 2 & 0 \\ 0 & 0 & 0 \end{bmatrix}.
\]
There is a $2$ on the main diagonal but it is tucked within a $2 \times 2$ block of $\mat{D}$.  Here in fact we have $N_0(q') = 6$ and $N_2(q') = 2$.  

An example of an alternating form $q''$ with $N_2(q'') > N_0(q'')$ is $q'' = 2x_1^2 + 2x_2^2 + 2 x_1 x_2$, which corresponds to a single edge with a self-loop at each end.  Replacing each self-loop by a triangle yields a 6-node simple undirected graph with $N_0 = 28$ and $N_2 = 36$.  We will show that when $N_0 \neq N_2$ in the alternating case, the absolute difference is a simple function of the rank $r$ of $\mat{B}$ over $\Eff_2$.  

\section{Main Results}\label{sec:rank}



Given any $n$-qubit stabilizer circuit $C$ of size $s$ with $h$ nondeterministic gates, we can obtain its associated quadratic form $q_C$ in $O(s)$ time via the process in Section~\ref{stabilizer}.  This form has variables $\vec{x} = x_1,\dots,x_n$ for inputs, $\vec{z} = z_1,\dots,z_n$ for outputs, and $y_1,\dots,y_h$ for nondeterministic variables (\emph{wlog.} all coming from $h$ Hadamard gates).  It may also have the variables called ``$w_j$'' in Section~\ref{stabilizer}, but those are introduced only to equate the final annotation term on a qubit line $j$ with the output variable $z_j$ without thereby forcing a value restriction for nondeterministic variable(s) on that line, and so preserve $2^{h/2}$ as the value of the magnitude divisor $R$ in (\ref{amplitude}).  We can either treat $w_j$ as forced by $z_j$ without changing $R$, or avoid introducing $w_j$ by reducing $R$.  Since the circuits are allowed to have initial $\gatem{X}$ gates on some lines, treating $\vec{x} = (0, \cdots, 0)$ loses no generality.  For any output $\vec{b} = (b_1, \cdots, b_m)$, the quadratic form then becomes
\begin{eqnarray*}
q(\vec{y},\vec{b}) &=& (\sum \alpha_i y_i + \sum 2 y_i y_j) + \sum 2 y_i b_j \mod{4}\\
&=& \vec{y}^\top \mat{A} \vec{y} + \vec{y}^\top 2 \mat{\Delta} \vec{y} \mod{4}
\end{eqnarray*}
in the $\vec{y}$ variables only.  
Here $\mat{\Delta}$ is a diagonal matrix with $\mat{\Delta}_{i,i} = b_j$.  Because we will have $h = \Theta(n)$ for the most part, we still refer to ``$n$'' to denote the number of variables in quadratic forms.  

Finally, we also fix the outputs $b_j$ all to be $0$.  We denote by $\vec{N} = (N_0,N_1,N_2,N_3)$ the resulting distribution of values of $q_C$ over the $2^h$ assignments to $\vec{y}$.  Reviewing the discussion surrounding \Cref{amplitude} in \Cref{stabilizer}, we can abbreviate the numerator of the amplitude by
\begin{equation}\label{a0}
a_0(\vec{N}) = N_0 - N_2 + i(N_1 - N_3).
\end{equation}
We use the $N_c$ and $a_0$ notation generally for linear and quadratic forms $f$ without reference to their coming from a quantum circuit.  Then $a_0$ gives the value of the exponential sum $\sum_x i^{f(x)}$.

Now the present the main lemmas that underlie the main theorems.  Their proofs are in the appendix.

\bigskip
\begin{lemma}\label{lm:basic}
For any linear function $f(x_1, \cdots, x_n) = \sum^n_{i=1} a_i x_i$ over $\Z_4$, $|N_0 - N_2|$ and $|N_1 - N_3|$ are 0 or a power of 2.
\end{lemma}

\bigskip
\begin{lemma}\label{lm:alter}
For any $\Z_4$-valued alternating quadratic form $f : V \rightarrow \Z_4$ of rank $r$, there is a basis of $V$ over which $f$ can be rewritten as
\[
f(\vec{x}) = 2 \sum^{g}_{j = 1} x_{2j - 1} x_{2j} + 2 \sum^{n}_{i = 1} w_i x_i
\]
for some $\vec{w} = (w_1, \cdots, w_n) \in K^n$, and
\[
N_0 - N_2 = 0\ \ \mathrm{or}\ \ (-1)^k 2^{n - g},
\]
where $2g = r$ and $k$ is the number of $(w_{2j-1}, w_{2j})$-pairs in $f$ such that $(w_{2j-1}, w_{2j}) = (1,1)$ for $j \in \{1, \cdots, g\}$.
Also $N_1 = N_3 = 0$.
\end{lemma}

\bigskip
\begin{lemma}\label{lm:nalter}
For any $\Z_4$-valued non-alternating quadratic form $f : V \rightarrow \Z_4$ of rank $r$, there is a basis of $V$ over which $f$ can be rewritten as
\[
f(\vec{x}) = \sum^{r}_{j = 1} x_j + 2 \sum^{n}_{i = 1} w_i x_i = \sum^r_{j=1} (1 + 2w_j) x_j + 2 \sum^n_{i = r+1} w_i x_i
\]
for some $\vec{w} = (w_1, \cdots, w_n)$. Define $c$ to be the number of $w_i$'s such that $w_i = 0$ with $i \in \{r+1, \cdots, n\}$ and $d$ to be the number of pairs such that $(1 + 2w_j, 1 + 2w_{j'}) = (1,3)$ with $j, j' \in \{1, \cdots, r\}$. Also let $m = n - c - 2d$ and rewrite $m = 4a + b$, and define $\eta$ such that $\eta = 0$ if the rest $m$-many coefficients are all $1$'s but $\eta = 1$ if they are all $3$'s. Then the differences $N_0 - N_2$ and $N_1 - N_3$ take one of the following values:
\begin{itemize}
  \item if $b = 0$, then $N_0 - N_2 = (-1)^a 2^{(n+c) / 2}$, $N_1 - N_3 = 0$;
  \item if $b = 1$, then $N_0 - N_2 = (-1)^a  2^{(n+c - 1) / 2}$, $N_1 - N_3 = (-1)^{a + \eta} 2^{(n+c - 1) / 2}$;
  \item if $b = 2$, then $N_0 - N_2 = 0$, $N_1 - N_3 = (-1)^{a +\eta} 2^{(n+c) / 2}$;
  \item if $b = 3$, then $N_0 - N_2 = (-1)^{a+1} 2^{(n+c - 1) / 2}$, $N_0 - N_2 = (-1)^{a+\eta} 2^{(n+c - 1) / 2}$.
\end{itemize}
\end{lemma}

The connection between rank and solution counting is expressed by our main theorem about quadratic forms after the normalization process in \Cref{lm:basic,lm:alter,lm:nalter} is applied:

\bigskip
\begin{theorem}\label{thm:reduction-to-rank}
Given any normalized classical quadratic form $f$ in $n$ variables, we can compute $N_0,N_1,N_2,N_3$ and hence $a_0(\vec{N})$ in time $O(n)$.  Furthermore, $|a_0(\vec{N})|^2$ is either $0$ or $2^{2n-r}$ where $r$ is the rank of $f$.
\end{theorem}



\bigskip
\noindent
This means that the bulk of the computing time for the whole process goes into the decomposition in Lemma~\ref{PLDLTPT}, which is used to compute the normal forms asserted in Corollary~\ref{cor:similar}.  After that, the up-to-$n^2$ denseness of the original form does not matter and the computation needs only $O(n)$ time.

\bigskip
\begin{proof}


We show this separately for the alternating and non-alternating cases. 
By \Cref{cor:similar}, a normalized alternating quadratic form is of the form
\[
f(\vec{x}) =  2 \sum^{r/2}_{j = 1} x_{2j - 1} x_{2j} + 2 \sum^{n}_{i = 1} w_i x_i \pmod{4},
\]
for some $\vec{w} = (w_1, \cdots, w_n) \in K^n$. It is easy to see that $N_1 - N_3$ is always zero since there is no assignment to $\vec{x} = (x_1, \cdots, x_n)$ that would give $f(\vec{x}) = 1$ or 3. \Cref{lm:alter} gives out
\[
N_0 - N_2 = 0\ \ \mathrm{or}\ \ (-1)^k 2^{n-g},
\]
which can be done in time $O(n)$.  Hence if this is non-zero, then we have
\[
a_0(\vec{N}) = (-1)^k 2^{n - g},
\]
and
\[
|a_0(\vec{N})|^2 = 2^{2n - r}.
\]

Similarly, a normalized non-alternating quadratic form is written as
\[
f(\vec{x}) = \sum^{r}_{j = 1} x_j + 2 \sum^{n}_{i = 1} w_i x_i = \sum^r_{j=1} (1 + 2w_j) x_j + 2 \sum^n_{i = r+1} w_i x_i \pmod{4},
\]
for some $\vec{w} = (w_1, \cdots, w_n) \in \{0,1\}^n$. Things become trivial if $2w_i = 2$ for some $i \in \{r+1, \cdots, n\}$. This makes $N_0 - N_2 = N_1 - N_3 = 0$.

Now assume $N_0 - N_2$ and $N_1 - N_3$ are not both zero at the same time. Then we can derive from \Cref{lm:nalter} that
\[
a_0(\vec{N}) = N_0 - N_2 + i (N_1 - N_3)
\]
takes one of the following values:
\begin{itemize}
  \item if $b = 0$, then $a_0(\vec{N}) = (-1)^a 2^{(n+c)/2}$;
  \item if $b = 1$, then $a_0(\vec{N}) = (-1)^a  2^{(n + c - 1) / 2} + i (-1)^{a + \eta} 2^{(n + c - 1) / 2}$;
  \item if $b = 2$, then $a_0(\vec{N}) = i (-1)^{a +\eta} 2^{(n+c)/2}$;;
  \item if $b = 3$, then $a_0(\vec{N}) = (-1)^{a+1} 2^{(n + c - 1) / 2} + i (-1)^{a+\eta} 2^{(n + c - 1) / 2}$,
\end{itemize}

where $a, b$ and $c$ are as defined in \Cref{lm:nalter}. Note that $c = n - r$. Together we have
%
%
\[
|a_0(\vec{N})|^2 = 2^{2n-r},
\]
and again this can be computed in $O(n)$ time.
\end{proof}

%
%
Now we rejoin the process of evaluating the stabilizer circuit $C$.  It will normalize $q_C$ to $q'$ in one of the two forms in \Cref{cor:similar}, which will give a matrix $\mat{D}'$ such that $q'(\vec{y}) = \vec{y}^\top \mat{D}' \vec{y}$. With such $\mat{D}'$, the acceptance probability can be derived directly by \Cref{thm:reduction-to-rank}. We repeat the statement of our main theorem from section~\ref{introduction} but split (b) into two pieces, proving part (b1) here and part (b2) in the next section.

\bigskip
\begin{theorem}[Main Theorem]\label{main}
\begin{itemize}
\item[(a)]
Strong simulation of $n$-qubit stabilizer circuits $C$ with $h$ nondeterministic single-qubit gates on standard-basis inputs (amplitude as well as the probability) is in time $O(s + n + h^{\omega})$ where $2 \leq \omega < 2.3729$.
\item[(b1)]
Computing $n \times n$ matrix rank over $\Eff_2$ reduces in linear time to computing one instance of the strong simulation probability $|\triplep{0^n}{C}{0^n}|^2$.
\item[(b2)]
Computing the strong simulation probability $p = |\triplep{0^n}{C}{0^n}|^2$ reduces in linear time to computing one instance of $n \times n$ matrix rank over $\Eff_2$ on the promise that $p > 0$.
\end{itemize}
\end{theorem}


\begin{proof}[Proof of (a) and (b1)]
(a)~Let $C$ be given, take $\mat{A}$ to be the matrix over $\Zed_4$ of its classical quadratic form $q_C$, and take $\mat{B}$ be the associated symmetric matrix over $\Eff_2$.  By Lemma~\ref{PLDLTPT} and the algorithm of \cite{DuPe18} there is a decomposition $\mat{B} = \mat{P}\mat{L}\mat{D}\mat{L}^\top \mat{P}^\top$ over $\Eff_2$ that is computable in $O(n^\omega)$ time
such that $\mat{D}$ is diagonal (in the non-alternating case) or $2 \times 2$ block-diagonal (in the alternating case) and equals the matrix $\mat{D}$ in Theorem~\ref{thm:similar}.  This also computes the rank $r$ of $\mat{B}$. Then compute $\mat{D}' = \mat{L}^{-1} \mat{P}^\top \mat{A} \mat{P}  (\mat{L}^{-1})^\top$ over $\Zed_4$ which again takes $O(n^\omega)$ time. By Lemma~\ref{PLDLTPT}(d), $\mat{D}'$ and $\mat{D}$ yield the vector $\vec{w}$ in the normal form of \Cref{cor:similar} for $q_C$. 
%
%
Then Theorem~\ref{thm:reduction-to-rank} yields not only the probability $p = |\triplep{0^n}{C}{0^n}|^2$ but also the entire distribution of phases as powers of $i$, and hence yield the amplitude $\triplep{0^n}{C}{0^n}$.

(b1)~To compute the rank $r$ of an $n \times n$ matrix over $\Eff_2$, make an equivalent symmetric matrix $\mat{A}$ by the block-transpose trick in the introduction.  Not only is $\mat{A}$ alternating but it is the adjacency matrix of a bipartite graph $G = (V,V',E)$.  To see that the corresponding graph state circuit $C$ gives $p = |\triplep{0^n}{C}{0^n}|^2 > 0$, consider any assignment $a$ to the variables in $V$.  This reduces $q_C$ to a linear form $2\ell(x')$ of the variables $x'$ corresponding to nodes of the other partition.  If $\ell(x')$ vanishes modulo 2, then all extensions of $a$ to $a'$ on $x'$ contribute 0 modulo 4.  Otherwise, $2\ell(x')$ has a nonzero term $2x'_i$ for some $i$.  Assignments $a'$ to $x'$ pair off with canceling contributions 0 and 2 according to the value $a'_i$ of $x'_i$.  Thus there are never more values of $2$ than $0$.  Finally, the all-zero assignment to $x$ makes $\ell(x')$ vanish, so the difference between the numbers of 0 values and 2 values is positive.  Thus the normal form for $q_C$ with input and output $0^n$ cannot have global cancellation, so $r$ is a simple function of $p$.
\end{proof}

\bigskip
To get the converse simulation in (b2) we must consider the non-alternating case, which arises when the stabilizer circuit $C$ has an odd number of $\gatem{S}$ or $\gatem{S^*}$ gates on some qubit line(s), and allow for the possibility $\triplep{0^n}{C}{0^n} = 0$.  The algorithm for amplitude in the non-alternating case needs knowledge of individual entries in the normal form over $\Zed_4$ besides the rank $r$ of $q_C$.  We show that for the probability computation $p = |\triplep{0^n}{C}{0^n}|^2$ one can reduce to the alternating case---that is, produce in $O(s + h^2)$ time an alternating quadratic form $q'_C$ of rank $2r$ such that $p$ is a simple function of $r$ provided $p > 0$.  This development leads us to a concept of ``self-dual'' quadratic forms in order to complete the proof of (b2) in the next section.

\section{Self-Dual Forms and Probability Reduction to Rank}\label{sec:self-dual}

To take stock, we have shown that strong simulation of stabilizer circuits (on input $0^n$) is in matrix-multiplication time, which is currently the best-known time for computing matrix rank (over $\Eff_2$).  We have also reduced the rank computation to the simulation of the restricted class of circuits that come from bipartite graphs.  The issue now becomes whether there is a more-general equivalence of strong simulation of stabilizer circuits $C$ to computing the rank.  We show a \emph{yes} answer provided the decision problem of telling whether the probability $p = |\triplep{0^n}{C}{0^n}|^2$ is nonzero, for $C$ having no phase gates, is in $O(n^2)$ time.

To do so, we begin by abstracting the quadratic forms $q + q^*$ in Theorem~\ref{probability} into a notion of ``self-duality.''  We do not need to legislate the separate presence of input variables $x_i$ and output variables $z_j$.  Our definition is at the level of the long literature of quadratic forms but we have not found a reference for it.

\bigskip
\begin{definition}\label{self-dual}
Let $\pi$ be a permutation of the variable set $V$ such that $\pi^2 = 1$ (that is, an involution).  Let $f(v)$ be a quadratic form in reduced form as a sum of terms over $\Zed_K$.  Then $f$ is \emph{self-dual} provided:
\begin{itemize}
\item
Whenever $f$ has the term $a v_i^2$, $v_i$ is not fixed by $\pi$ and $f$ also has the term $(K - a)\pi(v_i)^2$.
\item
Whenever $f$ has the term $b v_i v_j$, at least one of the variables is not fixed by $\Pi$ and $f$ also has the term $(K - b)\pi(v_i)\pi(v_j)$.
\end{itemize}
\end{definition}

\bigskip
\noindent
If $f$ had $av_i^2$ with $v_i$ fixed by $\pi$ then $f$ would also have $(K-a)v_i^2$ so the terms would cancel, and similarly if $b v_i v_j$ were fixed by $\pi$.  The polynomials $q(v) + q^*(v')$ in Theorem~\ref{probsim} meet this definition even before the substitution of values for $x_i$ and $z_j$ variables since those are considered fixed by the permutation sending $v$ to $v'$.  The definition allows $b$ to be odd in $b v_i v_j$ but classical forms with $K = 4$ only allow $b = 2$.  They can have terms of the form $v_i^2 + 3(v'_i)^2$ or $3v_i^2 + (v'_i)^2$, but our first results will eliminate them.

We can fix any partition of the set $V$ of non-fixed variables into $U$ and $U'$, writing $u' = \pi(u)$ for $u \in U$.  Also, given a Boolean assignment $a$, let $a' = \pi(a)$.

\bigskip
\begin{lemma}\label{self-dual-basics}
Given any self-dual quadratic form $f$ over $\Zed_4$:
\begin{itemize}
\item[(a)]
The form obtained by identifying any $u,u'$ pair is self-dual.
\item[(b)]
For all Boolean assignments $a$, $f(a) + f(a') = 0 \pmod{4}$.  Hence it follows that the numbers of $1$ values and $3$ values are equal, so their contributions cancel.
\item[(c)]
Values of $1$ and $3$ exist only when $f$ has terms $3u + u'$ or $u + 3u'$ where $u$ is in $U$, and when they happen, exactly half the values are $1$ and $3$.  Hence it suffices to count the number of zeroes of f.
\end{itemize}
\end{lemma}

\begin{proof}
Part (a) is immediate.  Note that if $f$ has terms $3u^2 + u'^2$ or $2xu + 2xu'$ with $x$ fixed this makes them cancel.  But $2uv + 2u'v'$ becomes $2uv + 2uv'$ which is not yet trivial.  This also compares to point~3 of Lemma 3.1 in \cite{CGW18}.  Part (b) follows individually for each term $t$ plus $t' = \pi(t)$.  For (c), let $T$ be the set of $u \in U$ such that $f$ has $3u^2 + u'^2$ or $u^2 + 3u'^2$.  Then an assignment $a$ makes $f(a)$ odd iff it sets an odd number of pairs $(u,u')$ with $u \in T$ to different values, and so exactly half the assignments do so.  On the other assignments, those terms contribute $0 \pmod{4}$ so those terms do not affect the calculation.
%
\end{proof}

\bigskip

Now we derive a combinatorial lemma that effectively eliminates the terms with odd coefficients.  Let $T_0 \subseteq U$ collect the variables $u$ such that $u^2 + 3u'^2$ is a term, and $T_1 \subseteq U$ those for which $3u^2 + u'^2$ is a term.  Now define
\[
S = \{a: V \to \{0,1\} : a(u) \neq a(u') \text{ for an even number of } u \in T_0 \cup T_1,\;\; a(u) = 1\}.
\]
Then $S$ is a linear subspace of $\Eff_2^{|V|}$.  Assignments $a$ outside $S$ make $f(a)$ have value 1 or 3 and hence contribute to a global cancellation by Theorem~\ref{self-dual-basics}(c).  We will now replace $f$ by an alternating form $f'$ that preserves the cancellation outside $S$ and gives the same distribution of $0$ and $2$ values inside $S$.  Let $t$ collect the terms in $f$ with odd coefficients and define a quadratic form $f'$ as follows:

\begin{enumerate}
\item
For every pair $(u,v)$ with $u,v$ both in $T_0$ or both in $T_1$, $r$ has the term-pair $2uv + 2u' v'$.
\item
For every pair $(u,v)$ with $u \in T_0$ and $v \in T_1$ or vice-versa, $r$ has the term-pair $2uv' + 2u' v$.
\item
Finally choose any variable $u \in T$ and substitute $2u$ by twice the sum of all the other variables in $T$ into $f - t + r$.  The result is $f'$.
\end{enumerate}

\noindent
We remark that in the second step, $r$ introduces terms that ``cross the partition'' from $U$ to $U'$.  The self-dual forms $q_C + q_C^*$ arising from Theorem~\ref{probability} have a canonical partition with no crossing edges.  The import here is that the presence of terms that cross the partition does not affect the analysis---they will just become edges in a larger graph.  Except for the element crossing the partition, $r$ represents exclusive-or-ing by a clique of edges between pairs of nodes in $T_0 \cup T_1$.  This also accounts for our current statement of the reduction time having a $+ n^2$ term that is not reduced when the size $s$ is $o(n^2)$.

The final form $f'$ in the third step need not even be self-dual.  For an example, consider $f = 3u + u' + 3v + v'$, which has no cross-terms.  Then $T = T_1 = \{u,v\}$.  So $r = 2uv + 2u'v'$.  The subspace $S$ is defined by $u + u' + v + v' = 0 \pmod{2}$.  Because $f - t + r = r$ only has even terms, we can substitute $2u + 2u' + 2v + 2v' = 0 \pmod{4}$.  Choosing $2v$ as the term to substitute for makes
\[
f' = u(2u + 2u' + 2v') + 2u'v' = 2u^2 + 2uu' + 2uv' + 2u'v'.
\]
This is not self-dual because of the absence of terms $2u'^2$, $2u'v$, and $2uv$, plus it has a term of the form $2uu'$, which even if it were allowed in Definition~\ref{self-dual} would be canceled in the process of forming $q + q^*$.  It is, however, alternating and equivalent to $f$ restricted to $S$ pointwise, not just in distribution.  That is, for any assignment $a$ to $(u,u',v')$, letting $a'(v) = a(u) + a(u') + a(v')$ (and $a'(u) = a(u)$, $a'(u') = a(u')$, $a'(v') = a(v')$), we get $f'(a) = f(a')$.  Now we prove this in general.

\bigskip
\begin{lemma}\label{nonalt-to-alt}
The classical alternating quadratic form $f'$ equals $f$ restricted to $S$, hence $N_0(f') - N_2(f') = N_0(f) - N_2(f)$.
\end{lemma}

\begin{proof}
The substitution step~3 is valid because $f - t + r$ only has even coefficients.  Hence we need only show that $f - t + r$ agrees with $f$ on $S$.  For any assignment $a: V \to \{0,1\}$, define $D$ to be the set of variables in $u \ T$ such that $a(u') = 1 - a(u)$.  Pick any $u \in D$, for instance the lexicographically first variable in $D$ .  Then define $a'$ to be the assignment with $a'(u) = 1 - a(u)$ and $a'(u') = 1 - a(u') = 1 - a'(u)$.  Going from $a$ to $a'$ changes the form $t$ by $+2$ modulo $4$. The only term-pairs in $r$ that change are $(2uv + 2u'v')$ or $(2uv' + 2u'v)$ for $v \in D$ and those change by $2$.  Hence if $d = |D|$ is even, there is a matching change by $2$, so $f'(a) - f'(a') = f(a) - f(a')$.  Define the closure of $a$ under such ``flip'' operations by $F(a)$.

For the $d$-even case, it remains to give some $a_0 \in F(a)$ such that $t(a_0) = r(a_0)$.  Take $a_0(u) = 1$, $a_0(u') = 0$ for $u \in T_0 \cap D$ and $a_0(u) = 0$, $a_0(u') = 1$ for $u \in T_1 \cap D$.  Then $t(a_0) \equiv d \pmod{4}$. To get the contribution from $r(a_0)$, set $k = |T_0 \cap D|$ and consider:

\begin{itemize}
\item
For every $v,w \in T \setminus D$, $r$ has the term-pair $2vw + 2v'w'$ or $2vw' + 2v'w$ where $a_0(v) = a_0(v')$ and $a_0(w) = a_0(w')$.  Each of these pairs contributes a multiple of $4$, so the net from these terms is $0$.
\item
For every pair $(v,w)$ with $v \in D$ and $w \in T \setminus D$, the term-pair equals $0$ if $a_0(w) = 0$ and $2(v+v') = 2$ otherwise.  The net for any $w$ with $a_0(w)  = 1$ is $2d$, which is $0$ if $d$ is even.
\item
For every pair with $v,w \in D$, the contribution is always $2$: Either $v$ and $w$ are both in $T_0$ or both in $T_1$ and the term pair is $2vw + 2v'w'$ or they have one in each and the pair is $2vw' + 2v'w$.  Hence the total contribution is $2\binom{d}{2} = d^2 - d$.
\end{itemize}

\noindent
The final point is that when $d$ is even, $d^2$ vanishes so the net from $r(a_0)$ is $-d$ modulo 4, which equals $d$ modulo $4$ from $t$.  This finishes the proof that for all assignments $a \in S$, $r(a) = t(a)$.
\end{proof}

\begin{proof}[Proof of Theorem~\ref{main}(b2)]
Given an $n$-qubit stabilizer circuit $C$ of size $s$ with $h$ nondeterministic gates, we can obtain the self-dual classical quadratic form $f = q_C + q^*_C$ in $O(s)$ time by the process in Theorem~\ref{probability}.  Then Lemma~\ref{nonalt-to-alt} converts $f$ to the equivalent alternating quadratic form $f'$ in time $O(s + h^2)$.  The alternating case of Theorem~\ref{thm:reduction-to-rank} then gives a simple relation from $p = |\triplep{0^n}{C}{0^n}|^2 = \frac{1}{R}|N_0(f') - N_2(f')|$ to rank unless $p = 0$.
\end{proof}

\medskip
\noindent
If $C$ has at most $k$ qubit lines with odd numbers of phase gates, where $k = O(\sqrt{s})$, then the time is $O(s)$.  The most immediate scope for improvement is to reduce the overhead from the quadratic size of $r$---perhaps sacrificing the exact agreement between $f'$ and $f$ for distributional equivalence.

\section{Conclusions}\label{conclusions}

We have improved the asymptotic running time for strong simulation of $n$-qubit stabilizer circuits (with typical size and nondeterminism) from $O(n^3)$ to $O(n^\omega)$.  We have also shown a linear time reduction from matrix rank over $\Eff_2$ to strong simulation.  One interpretation of the latter is:

\begin{fmquote}{5.25in}
The time gap between weak and strong simulation for stabilizer circuits cannot be closed unless $n \times n$ matrix rank over $\Eff_2$ is computable in $O(n^2)$ time.
\end{fmquote}

\noindent
The direction from the quantum simulation to matrix rank comes close to establishing a complete equivalence of them, especially for the simulation probability $p$.  Via analysis of ``self-dual'' forms we have reduced the probability computation to the alternating case, in which by Lemma~\ref{lm:alter} and Theorem~\ref{thm:reduction-to-rank} we get a simple expression whose absolute value depends only on the rank and whether $p = 0$.  That puts focus on the complexity of deciding whether $p$, or equivalently the amplitude $a = \triplep{0^n}{C}{0^n}$, is zero, specifically in the alternating case where $a$ is always real.  

The alternating case comes down to graph-state circuits $C_G$ and can be framed in terms apart from quantum computing.  Consider black/white two-colorings (not necessarily proper) of the $n$ vertices of $G$, and count the number of edges whose two nodes are both colored black.  Call those B-B edges.  Define $c_0$ to be the count of colorings that make an even number of B-B edges and $c_1 = 2^n - c_0$ to be the count of colorings that make an odd number of B-B edges.  The following is called $a(G)$ for ``amplitude'' and divides by $2^n$ not $2^{n/2}$ because $C_G$ has $2n$ Hadamard gates.
\[
a(G) = \frac{c_0 - c_1}{2^{n}}.
\]

\smallskip
\begin{definition}
Call an undirected graph $G$ \emph{net-zero} if $a = 0$,  \emph{net-positive} if $a > 0$, and \emph{net-negative} if $a < 0$.
\end{definition}

\bigskip
\noindent
The following proposition collects some basic facts:

\bigskip
\begin{proposition}\label{net-zero}
\begin{itemize}
\item[(a)]
Every odd cycle graph is net-zero.
\item[(b)]
Every bipartite graph is net-positive.
\item[(c)]
A graph is net-zero if and only if one of its connected components is net-zero.
\item[(d)]
If $G$ is net-zero, then the graph $G'$ obtained by attaching a new node $v$ only to one existing node $u$, then attaching a second new node $w$ only to $v$, is also net-zero.
\end{itemize}
\end{proposition}

\begin{proof}
Part (a) follows because every coloring has an even number of B-W edges.  Hence the number of monochrome edges is odd, and so complementing the coloring flips the parity between B-B and W-W edges.
Part (b) was part of the proof of Theorem~\ref{main}(b1).
Part (c) is intuitive from how the quantum state is a tensor product over the connected components, so the events of all-$0$ output on each component are independent.
The proof of (d) is that whether $u$ is colored black or white, exactly one of the four colorings of $v$ and $w$ creates one more B-B edge.  Thus $|\triplep{0^{n+2}}{C_{G'}}{0^{n+2}}|^2$ is directly proportional to $|\triplep{0^{n}}{C_{G}}{0^{n}}|^2$.  
\end{proof}

\smallskip
The smallest net-zero graph is the triangle graph.  The graph made by attaching a second triangle is net-zero, as is the graph made by attaching a triangle to any of the latter's four outer edges.  
As observed at the end of section~\ref{sec:properties}, the six-node graph consisting of two triangles connected by an edge is net-negative.  
Here are the connected net-zero graphs of $3$, $4$, and $5$ nodes:

\begin{center}
\includegraphics[width=5.5in]{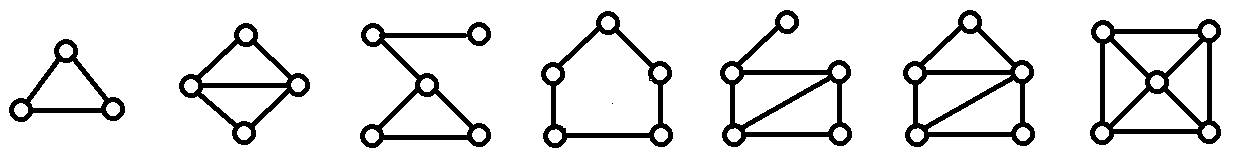}
\end{center}

\noindent
The concept extends to graphs with multiple edges and self-loops.  An isolated self-loop is net-zero, while an edge with two self-loops is net-negative.  This includes the quadratic forms produced by Lemma~\ref{nonalt-to-alt} and we conclude:

\bigskip
\begin{corollary}
If net-zero graphs of $n$ nodes with self-loops allowed are recognizable in $O(n^2)$ time then computing $|\triplep{0^n}{C}{0^n}|$ for stabilizer circuits $C$ (of $O(n^2)$ size with $O(n)$ nondeterminism) is $O(N)$-time equivalent to computing $n \times n$ matrix rank, where $N = n^2$.
\end{corollary}

\bigskip
The concept further extends to graphs with \emph{circles}, which are isolated loops without a vertex and contribute a multiplicative $-1$, and more generally to \emph{graphical 2-polymatroids} with \emph{rank function} $f_G(A)$ defined for any set $A \subseteq E$ to be the total number of vertices touched by edges in $A$.  Then $a(f_G)$ becomes a \emph{generalized Tutte invariant} (see \cite{OxWh93,Nob06}) with parameters
\[
(r,s,t;a,b,c,d;m,n) = (1/2,-1,0;1,-1,1,-1/2;1,-1/2).
\]
This gives
\[
a(G) = \left(-\frac{1}{2}\right)^{n/2} S(f_G; -\sqrt{2}i, \sqrt{2}i), \quad\text{where}\quad S(f;x,y) = \sum_{A \subseteq E} x^{f(E) - f(A)}y^{2|A| - f(A)},
\]
by the main theorem of \cite{OxWh93}.  This in turn further simplifies to
\[
a(G) = \sum_{A \subseteq E} \frac{(-2)^{|A|}}{2^{f_G(A)}}.
\]
Noble \cite{Nob06} shows that computing $S(f_G;x,y)$ is $\NumP$-hard for any constant rational $x,y$ whenever $xy \neq 1$.  The complex irrational point $(-\sqrt{2}i, \sqrt{2}i)$ has $xy = 2$ but evades his proof because having $y^2 = -2$ makes a denominator vanish.  
Other connections between quantum graph states and matroids have been shown by Sarvepalli \cite{Sar14}, and there is scope for further development along these lines.

\medskip
We have shown tight connections to the fundamental problems of counting solutions to quadratic forms $f$ over $\Eff_2$ and $\Zed_4$.  For $\Eff_2$ we get that $2f$ is an alternating form over $\Zed_4$ with the same solution count over $\{0,1\}^n$, so the near-equivalence to matrix rank applies.  In any event we have reduced the $\Zed_2$ case to matrix multiplication in a way that improves the $O(n^3)$ running time stated in \cite{EhKa90} to $O(n^\omega)$.  For binary solution counting of non-alternating classical quadratic forms over $\Zed_4$, we obtain $O(n^\omega)$ runtime via methods that multiply matrices as well as compute rank.  

When the non-Clifford gate $\gatem{CS}$ is added to create a universal set, the quadratic forms over $\Zed_4$ have terms $xy$ or $3xy$.  They are no longer classical and the connection to $\Eff_2$ exploited by $\cite{Schm09}$ no longer applies.  No such connection can apply, nor any extension of the algorithm in \cite{DuPe18} \emph{a-fortiori}, unless $\BQP = \Ptime$.  There is also the sharp dichotomy theorem of \cite{CLX14} that solution counting for these forms over all of $\Zed_4^n$ is in polynomial time, but over $\{0,1\}^n$ it is $\NumP$-complete.   This extends to affine versus non-affine forms over $\Zed_{K}^n$, $K = 2^k$.  Deeper understanding of \emph{why} the dichotomy operates may illuminate exactly which elements of quantum computations create hardness for classical emulation (for this, see \cite{Bac17,Bac18}).

Nevertheless, perhaps these techniques can apply to heuristic or approximative methods on general quantum circuits.
The polynomial translation in \cite{RCG18} applies to quantum circuits of all common gate types.  There are questions about analyzing circuits that are ``mostly Clifford'' or those from the Clifford plus $\gatem{T}$ libraries that try to minimize the latter gates, of which we mention\cite{BrGo16,MFIB18,BBCCGH19}.  For example, are there reasonably-tight bounds for the numbers of the non-Clifford gates required to compute certain functions that can be obtained efficiently by algebraic means, without resort to exhaustive search?

A direction for improving the present results is to sharpen the times for sparse cases---reflecting for instance the analysis for bounded degree in section VI of \cite{AnBr06} and the results for rank in \cite{CKL13}.  Our Lemma~\ref{nonalt-to-alt} introduces a dense clique of edges even for sparse graphs; perhaps there is a more economical reduction to the alternating case.  
We have left unused one further manipulation of a classical quadratic form $f$ that is most simply described in terms of the associated graphs $G$.  Subdivide each edge $e = (u,v)$ by a new node $s_e$ and add a second new node $r_e$ connected only to $s_e$.  Finally replace $2uv$ in the form by $2us_e + 2r_e s_e + 2vs_e$, and for each $u$ of odd degree (on non-self edges), add $2u^2$.  The resulting form $f'$ is equivalent to $f$ on the linear subspace $S$ of assignments that make $r_e = u + v \pmod{2}$ for each edge $e$ and is equivalent to a \emph{linear} form on that subspace.  The drawback is adding upwards of $n^2$-many nodes, but it preserves sparseness of the edge set.

A closer look into \Cref{lm:alter} and \Cref{lm:nalter} suggests that the probability of a specific output or the distribution over the entire output set can serve as a metric to test whether two given quantum stabilizer circuits are (not) \textit{equivalent}.
Let $h^i_0, h^i_1$ be two corresponding $h_0, h_1$ differences for circuit $C_i$/quadratic form $f_i(\vec{x})$. Then we can define the following two concepts accordingly.

\bigskip
\begin{definition}\label{dfn:weakdisE}
Given two quantum circuits $C_1$ and $C_2$, we call $C_1$ and $C_2$ are weakly equivalent, denoted by $C_1 \stackrel{w}{\approx} C_2$ if
\[
|\triplep{\vec{z}}{C_1}{\vec{a}}|^2 = |\triplep{\vec{z}}{C_2}{\vec{x}}|^2
\]
for a fixed input $\vec{a}$ and all possible output $\vec{z}$.
\end{definition}

\bigskip
\begin{definition}\label{dfn:disE}
Given two quantum circuits $C_1$ and $C_2$, we call $C_1$ and $C_2$ are strongly equivalent, denoted by $C_1 \stackrel{s}{\approx} C_2$ if for all possible output $\vec{z}$, the amplitudes of $C_1$ and $C_2$ are the same, that is,
\[
|N_j(Q_1(\vec{a}, \vec{y}, \vec{b}))| = |N_j(Q_2(\vec{a}, \vec{y}, \vec{b}))|
\]
for a fixed input $\vec{a}$ and all possible output $\vec{b}$.
\end{definition}

\bigskip
Now consider $C_1 \stackrel{s}{\approx} C_2$ for two given stabilizer circuits. The corresponding $Q_1(\vec{x}, \vec{y}, \vec{z})$ and $Q_2(\vec{x}, \vec{y}, \vec{z})$ will be of the forms as stated in \Cref{sec:rank}. Without loss of generality, assume $\vec{0}$. Note that the resulting $Q_1(\vec{y}, \vec{z})$ and $Q_2(\vec{y}, \vec{z})$ can be associated with two graphs. Each graph has two sets of nodes $\vec{y}$ and $\vec{z}$. The nodes in $\vec{y}$ can be connected by edges in any way, while there is no edge between nodes among $\vec{z}$ and each node is connected by exactly one node from $\vec{y}$ without overlapping node. Hence, this should be a strict class among graphs and this gives out another interesting question, \emph{does $C_1 \stackrel{s}{\approx} C_2$ implies that their associated graphs are isomorphic?} If this is true, we will have, $C_1 \stackrel{s}{\approx} C_2$ if and only if their associated graphs are isomorphic. Note that $C_1 \stackrel{s}{\approx} C_2$ says $|N_i(Q_1)| = |N_j(Q_2)|$ for all possible outputs, which are exponentially many. We ask:

\begin{itemize}
\item
If $|N_j(Q_1)| = |N_j(Q_2)|$ for all possible outputs, does $Q_1(\vec{y},\vec{b}) \stackrel{\F_2}{\sim} Q_2(\vec{y}, \vec{b})$ for all possible $\vec{b}$?
\item
For the case where $C_1 \stackrel{w}{\approx} C_2$, can we pose a similar question, but in terms of \emph{rank}?
\end{itemize}

\paragraph{Acknowledgments}
We thank Scott Aaronson, Graham Farr, Richard Lipton, and Virginia Vassilevska Williams for helpful information and discussions.

\bibliographystyle{alpha}
\bibliography{newpapers}


\newpage
\appendix
\section{Appendix: Other Results and Proofs}

%
%


\noindent
\textbf{Proof for \Cref{PLDLTPT}(c)}.  
Suppose $\mat{B}' = \mat{L}\mat{D}\mat{U} = \mat{M}\mat{E}\mat{V}$ where $\mat{L}\mat{M}$ are lower triangular and $\mat{U},\mat{V}$ are upper triangular, not even caring that $\mat{U} = \mat{L}^\top$ and $\mat{V} = \mat{M}^\top$ but just that they are invertible.

First consider the non-alternating case where $\mat{D}$ and $\mat{E}$ are diagonal but not necessarily of full rank.  They must have the same rank $r$.  Then $\mat{M}^{-1}$ is also lower triangular, so that $\mat{C} = \mat{M}^{-1}\mat{L}\mat{D}$ is lower triangular, and $\mat{U}^{-1}$ is upper triangular, so that $\mat{E}\mat{V}\mat{U}^{-1}$ is upper triangular.  $\mat{C} = \mat{M}\mat{L}\mat{D} = \mat{E}\mat{V}\mat{U}^{-1}$, and the only way a lower-triangular matrix can equal an upper-triangular matrix is when both are diagonal.  So $\mat{C}$ is diagonal, and we need only argue that $\mat{C} = \mat{D}$ (${} = \mat{E}$).  This follows because they have the same rank and for any $i$ such that $\mat{D}[i,i] = 0$, also $\mat{C}[i,i] = 0$.

In the alternating case, $\mat{M}^{-1}\mat{L}$ is lower triangular but its product $C$ with $\mat{D}$ can also have a non-zero diagonal above the main diagonal.  The product $\mat{E}\mat{V}\mat{U}^{-1}$ is upper-triangular except for the diagonal below the main.  Hence $C$ must be tri-diagonal.  Every off-diagonal nonzero element of $C$ equals a diagonal element of $\mat{M}^{-1}\mat{L}$ multiplied by the corresponding off-diagonal entry of $D$ and also equals a diagonal element of $\mat{V}\mat{U}^{-1}$ multiplying the corresponding entry of $\mat{E}$.  By invertibility over $\Eff_2$ the diagonal entries are all $1$, so we have proved that $D$ and $E$ agree on all off-diagonal entries.  The proof that they agree with each other (but not necessarily with $\mat{C}$) in their $1 \times 1$ blocks on the diagonal is similar to that for the alternating case.

\bigskip
\noindent \textbf{Proof for \Cref{lm:basic}}.
The given $f(x_1, \cdots, x_n)$ will fall into one of the following cases:

\begin{enumerate}
\item[(a)]
If some $a_j = 0$, it is safe to drop this $j$-th variable $x_j$ since $\sum^n_{i=1} a_i \cdot x_i \mod{4} = \sum^n_{i=1, i\neq j} a_i \cdot x_i \mod{4}$. Define $N'_0, N'_1, N'_2, N'_3$ with respect to $\vec{x}' = (x_1, \cdots, x_{j-1}, x_{j+1}, \cdots, x_n)$. We can see that $N_i = 2N'_i$ for $i = 0,1,2,3$;

\item[(b)]
If some $a_j = 2$, then for any $\vec{x}_0 = (x_1, \cdots, x_{j-1}, 0, x_{j+1}, \cdots, x_n)$, it can be paired with $\vec{x}_1$ such that $f(\vec{x}_1) = f(x_1, \cdots, x_{j-1}, 1, x_{j+1}, \cdots, x_n) = f(\vec{x}_0) + 2 \mod{4}$. That is, if $f(\vec{x}_0) = 0$, then $f(\vec{x}_1) = 2$, and vice versa. Same analysis goes to $N_1$ and $N_3$. Hence, the two differences are zero in this case;

\item[(c)]
If some $a_j = 1$ and some $a_k = 3$ (without loss of generality, assume $j \leq k$), then for any $\vec{x}_{10} = (x_1, \cdots, x_{j-1}, 1, x_{j+1}, \cdots, x_{k-1}, 0, x_{k+1}, \cdots, x_n)$ and $f(\vec{x}_{10}) = \sum^n_{i=1, i\neq j, i \neq k} a_i \cdot x_i + 1 \mod{4}$, we have $f(\vec{x}_{01}) = \sum^n_{i=1, i\neq j, i \neq k} a_i \cdot x_i + 3 \mod{4}$, which will cancel in the differences.

      While for $\vec{x}_{00} = (x_1, \cdots, x_{j-1}, 0, x_{j+1}, \cdots, x_{k-1}, 0, x_{k+1}, \cdots, x_n)$ and $f(\vec{x}_{00}) = \sum^n_{i=1, i\neq j, i \neq k}a_i \cdot x_i \mod{4}$, $f(\vec{x}_{11}) = \sum^n_{i=1, i\neq j, i \neq k} a_i \cdot x_i+ 4 \mod{4} = f(\vec{x}_{00})$. Hence, by dropping both $j$-th and $k$-th variables (similar to case 1) and defining $N'_i$ with respect to $\vec{x}' = (x_1, \cdots, x_{j-1}, x_{j+1}, \cdots, x_{k-1}, x_{k+1}, \cdots, x_n)$, $N_i = 2N'_i$ for $i = 0,1,2,3$;

\item[(d)]
If all $x_j$'s are 1, then for $i = 0, 1, 2, 3$, we have
      \[
      N_i = \sum_{m \geq 0} \binom{n}{4m + i}.
      \]
      Then \Cref{lm:4thBinomial} gives that both differences are powers of 2.

\item[(e)]
If all $x_j$'s are 3, then
      \[
      N_0 = \sum_{m \geq 0} \binom{n}{4m},\ \ \ \ N_2 = \sum_{m \geq 0} \binom{n}{4m+2},
      \]
      \[
      N_1 = \sum_{m \geq 0} \binom{n}{4m+3},\ \ \ N_3 = \sum_{m \geq 0} \binom{n}{4m + 1}.
      \]
      Then it can be reduced to case 4 and hence both differences are powers of 2.
\end{enumerate}
Note that the above procedures can be applied to a given $f(\vec{x})$ recursively. Overall, the statement holds. \qed

\bigskip
\begin{lemma}\label{lm:4thBinomial}
\[
\sum_{r \geq 0} \binom{n}{4r} - \sum_{r \geq 0} \binom{n}{4r+2}
\]
\[
\sum_{r \geq 0} \binom{n}{4r+1} - \sum_{r \geq 0} \binom{n}{4r+3}
\]
are either 0 or a power-of-2.
\end{lemma}

{\bf Proof.}It is known that with $\omega$ be the $d$-th root of unity,
\[
\sum_{r \geq 0} \binom{n}{d r + c} = \frac{1}{d} \sum^{d-1}_{j=0} \omega^{-jc} (1+ \omega^j)^n
\]
where $0 \leq c < d$. By simple substitutions with $d = 4$ and $a = 0, 1, 2, 3$, we get
\begin{align*}
g_0 &= \sum_{r \geq 0} \binom{n}{4 r} - \sum_{r \geq 0} \binom{n}{4 r + 2} = \frac{1}{2} (1+\omega^n)(1+\omega^3)^n, \\
g_1 &= \sum_{r \geq 0} \binom{n}{4 r+1} - \sum_{r \geq 0} \binom{n}{4 r + 3} = \frac{1}{2} \omega^{-1}(\omega^n - 1)(1+\omega^3)^n.
\end{align*}
Rewrite $n = 4a + b$ with $a,b \in \mathbb{Z}$ and $0 \leq b < 4$. Let $g_0 = \sum_{r \geq 0} \binom{n}{4 r} - \sum_{r \geq 0} \binom{n}{4 r + 2}$ and $g_1 = \sum_{r \geq 0} \binom{n}{4 r+1} - \sum_{r \geq 0} \binom{n}{4 r + 3}$. It is easy to verify that $(1+\omega^3)^4 = -4$ and hence we can rewrite
\begin{align*}
g_0 &= \frac{1}{2} (-4)^a (1+\omega^b)(1+\omega^3)^b, \\
g_1 &= \frac{1}{2} (-4)^a \omega^3 (\omega^b - 1)(1+\omega^3)^b.
\end{align*}
Now we can analysis them case by case.
\begin{enumerate}
  \item $b = 0$: $g_0 = (-4)^a$ and $g_1 = 0$;
  \item $b = 1$: $g_0 = (-4)^a$ and $g_1 = (-4)^a$;
  \item $b = 2$: $g_0 = 0$ and $g_1 = 2 \cdot (-4)^a$;
  \item $b = 3$: $g_0 = (-2) (-4)^a$ and $g_1 = 2 \cdot (-4)^a$.  \qed
\end{enumerate}

\bigskip
\noindent \textbf{Proof for \Cref{lm:alter}}.
Since $f$ is alternating, by \Cref{cor:similar} and \Cref{thm:similar}, $f$ has even rank $r$ and
\[
f(\vec{x}) = 2 \sum^{r/2}_{j = 1} x_{2j - 1} x_{2j} + 2 \sum^{n}_{i = 1} w_i x_i,
\]
for some basis for $V$ over $K$ and for some $\vec{w} = (w_1, \cdots, w_n) \in K^n$. Let $r = 2 g$ for some $g \in \Z$ and we can further rewrite it as
\[
f(\vec{x}) = 2 \sum^{g}_{j = 1} (x_{2j - 1} + w_{2j}) (x_{2j} + w_{2j - 1}) + 2 \sum^{n}_{i = r + 1} w_i x_i - 2 \sum^{g}_{j = 1} w_{2j - 1} w_{2j}
\]
Without loss of generality, we first look at the variable pair $(x_1, x_2)$ and its coefficient pair $(w_1, w_2)$. Denote
\[
f'(\vec{x}) = 2 \sum^{g}_{j = 2} (x_{2j - 1} + w_{2j}) (x_{2j} + w_{2j - 1}) + 2 \sum^{n}_{i = r + 1} w_i x_i - 2 \sum^{g}_{j = 2} w_{2j - 1} w_{2j},
\]
and write
\[
f(\vec{x}) = f'(\vec{x}) + 2 (x_1 + w_2) (x_2 + w_1) + 2 w_1 w_2.
\]
Note that $f'(\vec{x})$ only depends on $(x_3, \cdots, x_n)$, that is, $f'(\vec{x}) = f(x_3, \cdots, x_n)$. There are only four cases to consider for $h(x_1,x_2) = 2 (x_1 + w_2) (x_2 + w_1)$:
\begin{itemize}
  \item $(w_1, w_2) = (0,0)$: $h(x_1,x_2) = 2x_1 x_2$, and $h(0,0) = 0, h(0,1) = 0, h(1,0) = 0, h(1,1) = 2$;
  \item $(w_1, w_2) = (1,0)$: $h(x_1,x_2) = 2 x_1 (1 + x_2)$, and $h(0,0) = 0, h(0,1) = 0, h(1,0) = 2, h(1,1) = 0$;
  \item $(w_1, w_2) = (0,1)$: $h(x_1,x_2) = 2 (1 + x_1) x_2$, and $h(0,0) = 0, h(0,1) = 2, h(1,0) = 0, h(1,1) = 0$;
  \item $(w_1, w_2) = (1,1)$: $h(x_1,x_2) = 2(1 + x_1)(1 + x_2)$, and $h(0,0) = 2, h(0,1) = 0, h(1,0) = 0, h(1,1) = 0$.
\end{itemize}
Define $Q^{00}_i = \{ \vec{x} \in V | f(\vec{x}) = i \mod{4}\ \mathrm{and}\ x_1 = 0, x_2 = 0\}$ and similarly $Q^{01}_i, Q^{10}_i , Q^{11}_i$. Then we have $Q_i = Q^{00}_i \cup Q^{01}_i \cup Q^{10}_i \cup Q^{11}_i$. Also define $S^{00}_i = \{ \vec{x} \in V | f'(\vec{x}) = i \mod{4}\ \mathrm{and}\ x_1 = 0, x_2 = 0\}$ and analogously $S^{01}_i, S^{10}_i, S^{11}_i$.

Note that $f'(\vec{x})$ only depends on $(x_3, \cdots, x_n)$. Let $S'_i = \{ (x_3, \cdots, x_n) | f'(x_3, \cdots, x_n) = i\}$, and we have $|S'_i| = |S^{00}_i| = |S^{01}_i| = |S^{10}_i| = |S^{11}_i|$. Now analyze the above four cases separately:
\begin{itemize}
  \item $(w_1, w_2) = (0,0)$: we have
    \[
    f(\vec{x}) = f'(\vec{x}) + 2x_1 x_2,
    \]
    If $\vec{x}_{00} \in Q^{00}_i$, then $f(\vec{x}_{00}) = f'(\vec{x}_{00}) = i$, same for $Q^{01}_i, Q^{10}_i$, while if $\vec{x}_{11} \in Q^{11}_i$, then $f(\vec{x}_{11}) = f'(\vec{x}_{11}) + 2$ and hence $f'(\vec{x}_{11}) = i + 2$. Now for some $c \in \{0, 1\}$,
    \begin{align*}
    N_c - N_{c+2} &= |Q^{00}_c| + |Q^{01}_c| + |Q^{10}_c| + |Q^{11}_c| - (|Q^{00}_{c+2}| + |Q^{01}_{c+2}| + |Q^{10}_{c+2}| + |Q^{11}_{c+2}|)\\
    &= |S^{00}_c| + |S^{01}_c| + |S^{10}_c| + |S^{11}_{c+2}| - (|S^{00}_{c+2}| + |S^{01}_{c+2}| + |S^{10}_{c+2}| + |S^{11}_c|)\\
    &= 3|S'_c| + |S'_{c+2}| - (3|S'_{c+2}| + |S'_c|) \\
    &= 2(|S'_c| - |S'_{c+2}|).
    \end{align*}
  \item $(w_1, w_2) = (1,0)$: by the similar analysis, $|Q_c| - |Q_{c + 2}| = 2(|S'_c| - |S'_{c+2}|)$.
  \item $(w_1, w_2) = (0,1)$: $|Q_c| - |Q_{c + 2}| = 2(|S'_c| - |S'_{c+2}|)$.
  \item $(w_1, w_2) = (1,1)$: $|Q_c| - |Q_{c + 2}| = -2(|S'_c| - |S'_{c+2}|)$.
\end{itemize}
Hence, we can reduce the counting of $|Q_c| - |Q_{c + 2}|$ over $(x_1, \cdots, x_n)$ to the counting of $|S'_c| - |S'_{c + 2}|$ over $(x_3, \cdots, x_n)$, and gradually after $g$-many such reduction, we can derive
\[
|Q_c| - |Q_{c + 2}| = (-1)^m 2^g (|Q'_c| - |Q'_{c + 2}|),
\]
where $Q'_c = \{ (x_{r+1}, \cdots, x_n) | 2\sum^n_{i = r+1} w_i x_i = c\}$ and $m$ is the number of $(w_{2j-1}, w_{2j})$ pairs in $f(\vec{x})$ such that $(w_{2j-1}, w_{2j}) = (1,1)$, hence $2w_{2j-1} w_{2j} = 2$.

Now it is left to argue that $\left| |Q'_c| - |Q'_{c + 2}| \right|$ is either zero or a power of 2. Let $q(x) = \sum^n_{i = r+1} 2 w_i x_i$. Since $w_i \in \{0,1\}$, $q(x)$ is linear with coefficient from $\{0, 2\}$. Then we can reduce it to the 1st and 2nd cases in \Cref{lm:basic}. In the 2nd case, it gives that $|Q'_c| - |Q'_{c + 2}| = 0$ if $w_i = 1$ for some $i \in \{r+1, \cdots, n\}$. Now assume non-zero case. Then we have $w_i = 0$ for all $i \in \{r+1, \cdots, n\}$, which gives
\[
|Q'_c| - |Q'_{c + 2}| = (-1)^m 2^g 2^{n - r} = (-1)^m 2^{n - g},
\]
and hence it completes the proof. \qed

\bigskip
\noindent \textbf{Proof for \Cref{lm:nalter}}.
Since $f$ is non-alternating, by \Cref{cor:similar}, there exists a basis for $V$ over $K$, determining the coordinates $(x_1, \cdots, x_n)$, such that
\[
f(\vec{x}) = \sum^{r}_{j = 1} x_j + 2 \sum^{n}_{i = 1} w_i x_i,
\]
for some $\vec{w} = (w_1, \cdots, w_n) \in K^n$. By rearranging, we have
\[
f(\vec{x}) = \sum^r_{j=1} (1 + 2w_j) x_j + 2 \sum^n_{i=r+1} w_i x_i = \sum^r_{j=1} w'_j x_j + 2 \sum^n_{i=r+1} w_i x_i,
\]
where $w'_j = 1 + 2w_j$. Note that $w'_j$ can only be 1 or 3. Then we can reduce it to the 2nd, 3rd, 4th and 5th cases in \Cref{lm:basic}.

The 2nd case gives the trivial case where both $N_0 - N_2$ and $N_1 - N_3$ are zero. Now assume non-zero case. Then we have $w'_i, w_i \in \{0, 1, 3\}$.

Define $c$ to be the number of $w_i$'s such that $w_i = 0$ with $i \in \{r+1, \cdots, n\}$ and $d$ to be the number of pairs such that $(1 + 2w_j, 1 + 2w_{j'}) = (1,3)$ with $j, j' \in \{1, \cdots, r\}$. Also let $m = n - c - 2d$ and rewrite $m = 4a + b$, and define $\eta$ such that $\eta = 0$ if the rest $m$-many coefficients are all $1$'s but $\eta = 1$ if they are all $3$'s. Then the differences $N_0 - N_2$ and $N_1 - N_3$ are taking one of the following values:
\begin{itemize}
  \item if $b = 0$, then $N_0 - N_2 = (-1)^a 2^{(n+c) / 2}$, $N_1 - N_3 = 0$;
  \item if $b = 1$, then $N_0 - N_2 = (-1)^a  2^{(n+c - 1) / 2}$, $N_1 - N_3 = (-1)^{a + \eta} 2^{(n+c - 1) / 2}$;
  \item if $b = 2$, then $N_0 - N_2 = 0$, $N_1 - N_3 = (-1)^{a +\eta} 2^{(n+c) / 2}$;;
  \item if $b = 3$, then $N_0 - N_2 = (-1)^{a+1} 2^{(n+c - 1) / 2}$, $N_0 - N_2 = (-1)^{a+\eta} 2^{(n+c - 1) / 2}$.  $\Box$
\end{itemize}

\bigskip
Note that
Lemmas~\ref{lm:basic}, \ref{lm:alter}, and \ref{lm:nalter} and the proof method of Theorem~\ref{thm:reduction-to-rank} apply
to more general input $\vec{a}$ and output $\vec{b}$ as well, so that we have the following supplementary result:


\bigskip
\begin{theorem}\label{general-result}
Given a stabilizer circuit $C$ and its quadratic form $q_C(\vec{y},\vec{z})$, assume we know $\mat{Q},\mat{D}_1$ and $\mat{D}_2$ with entries in $\F_2$ such that $\vec{y}^\top \mat{Q}^\top \mat{A} \mat{Q} \vec{y} = \vec{y}^\top  (\mat{D_1} + 2\mat{D}_2) \vec{y}$ where
\begin{itemize}
  \item if $q_C$ is alternating, $\mat{D}_2$ is a diagonal matrix with entries in $\{0,1\}$ and $\mat{D}_1 = \mat{M}_1 \oplus \cdots \oplus \mat{M}_g$ has even rank $r = 2g$ over $\F_2$;
  \item if $q_C$ is non-alternating, $\mat{D}_1$ and $\mat{D}_2$ are both diagonal matrices with entries in $\{0,1\}$.
\end{itemize}
Then we can compute $|\triplep{\vec{b}}{C}{\vec{0}}|^2$ for any output vector $\vec{b}$ to the circuit
in $O(e n)$ time where $n = |\vec{y}|$ and $e$ is the number of ones in $\vec{y}$.
\end{theorem}

\begin{proof}
Assume $\mat{Q} = (Q_{i,j})$ with $Q_{i,j} \in \F_2$ and take any output vector $\vec{b}$. Then $q_C(\vec{y},\vec{b}) = \vec{y}^\top \mat{A} \vec{y} + \vec{y}^\top \mat{\Delta} \vec{y}$ and we have
\begin{align*}
\vec{y}^\top \mat{Q}^\top \mat{A} \mat{Q} \vec{y} + \vec{y}^\top \mat{Q}^\top \mat{\Delta} \mat{Q} \vec{y} &= \vec{y}^\top  (\mat{D_1} + 2 \mat{D}_2) \vec{y} + \sum_i 2 \vec{y}^\top \mat{E}_i \vec{y} \\
&= \vec{y}^\top  \mat{D_1} \vec{y} + \vec{y}^\top 2(\mat{D}_2 + \sum_i \mat{E}_i) \vec{y}
\end{align*}
where $E_i$ is a diagonal matrix $\diag(Q_{i,1}, \cdots, Q_{i,n})$ for $i$ such that $\Delta_{i,i} = 1$ and $\mat{D}_1$ varies depending on whether it is alternating or non-alternating. Then each $E_i = \diag(Q_{i,1}, \cdots, Q_{i,n})$ can be obtained in $O(n)$ time given the matrix $\mat{Q}$.

We also know that in both the alternating and non-alternating cases, the output probability $|\triplep{\vec{b}}{C}{\vec{0}}|^2$ is determined by the rank of $\mat{D}_1$ if $|\triplep{\vec{b}}{C}{\vec{0}}|^2 \neq 0$. Now we will show that with the knowledge of such $\mat{Q}$, we can tell $|\triplep{\vec{b}}{C}{\vec{0}}|^2 = 0$ in $O(e n)$ time.

First suppose $q_C(\vec{y},\vec{z})$ is alternating and $n = |\vec{y}|$, then for output $\vec{b}$ we can rewrite
\[
\vec{y}^\top  \mat{D_1} \vec{y} + \vec{y}^\top 2(\mat{D}_2 + \sum_i \mat{E}_i) \vec{y} = \sum^g_{j=1} 2 y_{2j-1} y_{2j} + \sum^n_{i=1} 2 w_i y_i \mod{4},
\]
where $w_i \in \{0,1\}$.
Once we finish updating the above equation (which takes $O(e n)$ time), we can by \Cref{lm:alter}, get the value $\braket{\vec{b}|C|\vec{0}}$ and identify if $|\braket{\vec{b}|C|\vec{0}}|^2 = 0$ which happens when $w_i = 0$ for some $i \in \{r+1, \cdots, n\}$. Analogously, this also can be done in $O(e n)$ time by \Cref{lm:nalter} for non-alternating cases.
\end{proof}

\end{document}